\documentclass{article}

\usepackage[utf8]{inputenc}
\usepackage[margin=1in]{geometry}
\usepackage{authblk}
\usepackage{mdframed}
\usepackage{amsmath}
\usepackage{amsthm}
\usepackage{amssymb}
\usepackage{thmtools}
\usepackage{todonotes}
\usepackage{thm-restate}
\usepackage{enumitem}
\usepackage{dsfont}
\usepackage{xspace}
\usepackage{comment}
\usepackage{xparse}
\usepackage{tikz}
\usepackage{fixme}
\usepackage[colorlinks]{hyperref}
\usepackage[capitalize]{cleveref}
\usepackage{subcaption}
\usepackage{soul}
\usepackage{xcolor}
\usepackage{xspace}
\usepackage{xfrac}

\usepackage{algorithm}
\usepackage[noend]{algpseudocode}

\usepackage{mathabx}
\usepackage{nicefrac}
\newtheorem{theorem}{Theorem}
\newtheorem{claim}[theorem]{Claim}
\newtheorem{corollary}[theorem]{Corollary}
\newtheorem{lemma}[theorem]{Lemma}
\newtheorem{proposition}[theorem]{Proposition}
\newtheorem{fact}[theorem]{Fact}

\newcommand{\Sk}{{\mathcal{S}}}

\newcommand{\clique}[1]{C_{#1}}

\usepackage{graphicx}
\usepackage{optidef}
\usepackage{titlesec}

\graphicspath{{./figures/}}

\allowdisplaybreaks[1]

\crefname{observation}{Observation}{Observations}
\crefname{algocfline}{Line}{Lines}

\newcommand{\lo}{$\ell_0$ Best-Fit Ultrametrics}

\newcommand{\set}[1]{\{#1\}}

\newcommand{\calC}{\mathcal{C}}

\titlespacing*{\subparagraph}{\parindent}{*0.5}{*0.5}

\date{}

\newtheorem{definition}{Definition}[section]

\title{\Large Fitting Tree Metrics and Ultrametrics in Data Streams}
\author{
Amir Carmel\thanks{Pennsylvania State University, United States. Part of this work was done while the author was affiliated at Weizmann Institute of Science, Israel. \texttt{amir6423@gmail.com}}
\qquad
Debarati Das\thanks{Pennsylvania State University, United States. \texttt{debaratix710@gmail.com}} 
\qquad
Evangelos Kipouridis \thanks{Max Planck Institute for Informatics, Saarland Informatics Campus, Saarbr{\"u}cken, Germany. \texttt{kipouridis@mpi-inf.mpg.de}}
\qquad
Evangelos Pipis\thanks{National Technical University of Athens, Greece, and Max Planck Institute for Informatics, Saarland Informatics Campus, Saarbr{\"u}cken, Germany. \texttt{evpipis@gmail.com}}
}

\begin{document}

\maketitle

\renewcommand{\thefootnote}{}
\footnotetext{Funding: This work is supported by NSF grant 2337832.}
\renewcommand{\thefootnote}{\arabic{footnote}}

\vspace{-1cm}

\begin{abstract}

    Fitting distances to tree metrics and ultrametrics are two widely used methods in hierarchical clustering, primarily explored within the context of numerical taxonomy. Formally, given a positive distance function $ D: \binom{V}{2} \rightarrow \mathbb{R}_{>0} $, the goal is to find a tree (or an ultrametric) $ T $ including all elements of set $ V $, such that the difference between the distances among vertices in $ T $ and those specified by $ D $ is minimized. Numerical taxonomy was first introduced by Sneath and Sokal [Nature 1962], and since then it has been studied extensively in both biology and computer science.

    In this paper, we initiate the study of ultrametric and tree metric fitting problems in the semi-streaming model, where the distances between pairs of elements from $ V $ (with $|V|=n$), defined by the function $ D $, can arrive in an arbitrary order. We study these problems under various distance norms; namely the $\ell_0$ objective, which aims to minimize the number of modified entries in $ D $ to fit a tree-metric or an ultrametric; the $\ell_1$ objective, which seeks to minimize the total sum of distance errors across all pairs of points in $ V $; and the $\ell_\infty$ objective, which focuses on minimizing the maximum error incurred by any entries in $ D $.

\begin{itemize}
    \item Our first result addresses the $\ell_0$ objective. We provide a single-pass polynomial-time $\tilde{O}(n)$-space $O(1)$ approximation algorithm for ultrametrics and prove that no single-pass exact algorithm exists, even with exponential time.
    
    \item Next, we show that the algorithm for $\ell_0$ implies an $O(\Delta/\delta)$ approximation for the $\ell_1$ objective, where $\Delta$ is the maximum, and $\delta$ is the minimum absolute difference between distances in the input. This bound matches the best-known approximation for the RAM model using a combinatorial algorithm when $\Delta/\delta = O(n)$.
    
    \item For the $\ell_\infty$ objective, we provide a complete characterization of the ultrametric fitting problem. First, we present a single-pass polynomial-time $\tilde{O}(n)$-space 2-approximation algorithm and show that no better than 2-approximation is possible, even with exponential time. Furthermore, we show that with an additional pass, it is possible to achieve a polynomial-time exact algorithm for ultrametrics.
    
    \item Finally, we extend all these results to tree metrics by using only one additional pass through the stream and without asymptotically increasing the approximation factor.

\end{itemize}

\end{abstract}

\newpage
\pagenumbering{arabic}
\newcommand{\abs}[1]{\left|#1\right|}
\newcommand{\agreements}[1]{A(#1)}
\newcommand{\opttree}{\mathcal{T}_{ALG}}
\newcommand{\close}{\texttt{close}}

\section{Introduction} 
Hierarchical clustering is a method of cluster analysis that builds a hierarchy of clusters by starting with each data point as its own cluster and successively merging the two closest clusters until all points are merged into a single cluster or a stopping criterion is met. This method involves creating a \emph{dendrogram}, a tree-like diagram, that records the sequence of merges or splits, allowing for easy visualization and interpretation of the hierarchical structure. Hierarchical clustering uses various distance metrics (e.g., Euclidean, Manhattan, cosine) and linkage criteria (e.g., single, complete, average, Ward’s method), providing flexibility to tailor the analysis to specific data characteristics and clustering goals. It is versatile across different types of data, including numerical, categorical, and binary data, and has become the preferred method for analyzing gene expression data~\cite{D05} and constructing phylogenetic trees~\cite{charikar, R1}. Consequently, hierarchical clustering plays a significant role in both theory and practice across various domains, such as image processing to group similar regions within images~\cite{1395986}, social network analysis to identify communities within a network~\cite{BREIGER1975328}, and business and marketing to segment customers based on behavior, preferences, or purchasing patterns~\cite{KUMAR2020126}.

Tree metrics and ultrametrics are fundamental measures used in hierarchical clustering, where the distance between any two points is defined by the cost of the unique path connecting them in a tree-like structure. Formally, given a distance function \( D: \binom{V}{2} \rightarrow \mathbb{R}_{>0} \), the goal is to find a tree \( T \) with positive edge weights, encompassing all elements of set \( V \) as vertices. This tree \( T \) should best match the distances specified by \( D \). In the case of ultrametrics, the tree must be rooted, and all elements of $V$ must appear as leaf nodes at the same depth.

The task of fitting distances with tree metrics and ultrametrics, often referred to as the numerical taxonomy problem, has been a subject of interest since the 1960s~\cite{4-cavalli67, sneath1962numerical, 7-Sneath-Numerical-1963}. One of the pioneering works in this area was presented by Cavalli-Sforza and Edwards in 1967~\cite{4-cavalli67}. Different formulations of the optimal fit for a given distance function \( D \) lead to various objectives, such as minimizing the number of disagreements (using the \( \ell_0 \) norm of the error vector), minimizing the sum of differences (using the \( \ell_1 \) norm), and minimizing the maximum error (using the \( \ell_\infty \) norm).

Deploying hierarchical clustering (HC) algorithms faces significant challenges due to scalability issues, particularly with the rise of data-intensive applications and evolving datasets. As data volumes continue to grow, there is an urgent need for efficient algorithms tailored for large-scale models such as streaming, local algorithms, MPC, and dynamic models, given the large input sizes relative to available resources. In this work we study hierarchical clustering, focusing on tree metrics and ultrametrics in the semi-streaming model.
The model supports incremental updates, keeping the information about the clusters current without the need to reprocess the entire dataset.
This adaptability makes hierarchical clustering highly valuable for applications such as network monitoring and social media analysis, where real-time insights are essential~\cite{rodrigues2008hierarchical, lee2014incremental,luhr2009incremental}.

A recent result~\cite{AssadiCLMW22} studied hierarchical clustering (HC) in the graph streaming model, providing a polynomial-time, single-pass \(\tilde{O}(n)\) space algorithm that achieves an \(O(\sqrt{\log n})\) approximation for HC. When space is more limited, specifically \(n^{1-o(1)}\), the authors show that no algorithm can estimate the value of the optimal hierarchical tree within an \(o(\log n \log \log n)\) factor, even with \(poly\log n\) passes over the input and exponential time. 

A special case of the ultrametric fitting problem is where the tree depth is two, known as the \emph{Correlation Clustering problem}. 
In this problem given a complete graph $G=(V, E)$ with edges labeled either similar (0) or dissimilar (1), the objective is to partition the vertices $V$ into clusters to minimize the disagreements. 
After a decade of extensive research on correlation clustering in the semi-streaming setting~\cite{ChierichettiDK14, ahn2021correlation, cohen2021correlation, Assadi022, BehnezhadCMT22, Cohen-AddadLMP22, BehnezhadCMT23, Cohen-AddadLPTY24},
a recent breakthrough in~\cite{Cohen-AddadLPTY24} introduces
a single-pass algorithm that achieves a $ 1.847$ approximation using \(\tilde{O}(n)\) space.
This directly improves upon two independent works~\cite{CambusKLPU24, chakrabarty2023single}, both presenting single-pass algorithms achieving a $(3+\varepsilon)$-approximation using \(O(n/\varepsilon)\) space.
 
However, our understanding of streaming algorithms for larger depths, particularly within the context of ultrametrics and tree metrics, is very limited.
The challenge arises from the fact that, unlike correlation clustering, which deals with only two distinct input distances, this problem may involve up to \(n^2\) different distances, especially in a highly noisy input.
Although the optimal output tree can be defined using at most $n$ of these $n^2$ distances, identifying these $n$ distances is non-trivial. As a result, in the worst case, it may be necessary to store all observed input distances, which would require quadratic space if done naively.
Additionally, the hierarchical nature of clusters at various tree depths introduces inherent dependencies among clusters at different levels. This complexity makes it highly challenging to adapt the ideas used in streaming algorithms for correlation clustering to ultrametrics and tree metrics.
In this paper, we offer the first theoretical guarantees by providing several algorithmic results for fitting distances using ultrametrics and tree metrics in the semi-streaming setting under various distance norms.

\subsection{Other Related Works}\label{section:otherrealted}

\paragraph{Ultrametrics and Tree metrics.} The numerical taxonomy problem, which involves fitting distances with tree metrics and ultrametrics, was first introduced by Cavalli-Sforza and Edwards in 1967~\cite{4-cavalli67}. Day demonstrated that this problem is NP-hard for both $\ell_1$ and $\ell_2$ norms in the context of tree metrics and ultrametrics~\cite{day}. Moreover, these problems are APX-hard~\cite{apxCorrClust}, as inferred from the APX-hardness of Correlation Clustering, which rules out the possibility of a polynomial-time approximation scheme. On the algorithmic side, Harp, Kannan, and McGregor~\cite{mcgregor} developed an $O(\min\{n, k \log{n}\}^{1/p})$ approximation for the $\ell_p$ objective in the ultrametric fitting problem, where $k$ is the number of distinct distances in the input. Ailon and Charikar~\cite{charikar} improved this to an $O(((\log{n})(\log\log{n}))^{1/p})$ approximation, which they extended to the tree metric using a reduction from Agarwala~\cite{agarwala}.
A recent breakthrough in~\cite{debarati} presented the first constant-factor approximation for the $\ell_1$ objective for both ultrametric and tree metric. 

The $\ell_0$ objective was first investigated in~\cite{cohen2022fitting}, which developed a novel constant-factor approximation algorithm.
Charikar and Gao~\cite{CharikarG24} improved the approximation guarantee to 5. For the weighted ultrametric violation distance, where the weights satisfy the triangle inequality, they provided an $O(\min\{L, \log n\})$ approximation, with $L$ being the number of distinct values in the input. Kipouridis~\cite{kipouridis2023fitting} further extended these results to tree metrics.

Research into the $\ell_\infty$ numerical taxonomy began in the early 1990s.
It was discovered by several authors that the $\ell_\infty$ case of the ultrametric fitting problem is solvable in polynomial time (in fact linear time in the input) and it is the only case with that property, whereas the problem of $\ell_\infty$ tree fitting is APX-hard~\cite{kvrivanek1988complexity, chepoi2000approximation, farach1993robust, agarwala, wareham1992computational}. Since then, the 
$\ell_\infty$ Best-Fit Ultrametrics/Tree-Metrics problems were extensively studied from both mathematical and computational perspectives~\cite{chepoi2000approximation, bernstein2017infinity, bernstein2020infinity, ardila2005subdominant, dress2005delta, ma1999fitting,ca20,CAVL21}.

\paragraph{Correlation Clustering.} The classic correlation clustering problem, introduced by Bansal, Blum, and Chawla~\cite{DBLP:conf/focs/BansalBC02}, can be visualized as a special case of ultrametrics where the tree's depth is bounded by two. Correlation clustering serves as a fundamental building block for constructing ultrametrics and tree metrics. Despite being APX-hard~\cite{apxCorrClust}, extensive research~\cite{DBLP:conf/focs/BansalBC02, apxCorrClust, ChawlaMSY15, Cohen-AddadLN22,Cohen-AddadL0N23} has aimed at finding efficient approximation algorithms, with the latest being a 1.437-approximation~\cite{DBLP:conf/stoc/CaoCL0NV24}. Correlation clustering also boasts a rich body of literature and has been extensively studied across various models designed for large datasets, including streaming~\cite{ahn2021correlation, Assadi022, BehnezhadCMT22,Cohen-AddadLN22}, MPC~\cite{cohen2021correlation}, MapReduce~\cite{ChierichettiDK14}, and dynamic models~\cite{BehnezhadDHSS19,abs-2404-06797,abs-2402-15668}.

\paragraph{Metric Violation Distance.} Another counterpart of the ultrametric violation distance problem is the metric violation distance problem, which requires embedding an arbitrary distance matrix into a metric space while minimizing the $\ell_0$ objective. While only a hardness of approximation of 2 is known, ~\cite{GilbertJ17, FanRB18,GilbertGRRW20} provided algorithms with an approximation ratio of $O(OPT^{1/3})$.
An exponential improvement in the approximation guarantee to \(O(\log n)\) was achieved in~\cite{cohen2022fitting}. The maximization version of this problem is also well-motivated by its application in designing metric filtration schemes for analyzing chromosome structures, as studied in~\cite{chromosomeCorrClust}.

\subsection{Our Contributions}

In this work, we examine the problem of fitting tree metrics and ultrametrics in the semi-streaming model, focusing on the \(\ell_0\) and \(\ell_\infty\) objectives. 
Note that storing the tree alone requires \(\Omega(n)\) word space. Since we are working within the semi-streaming model, where \(\tilde{O}(n)\) space is permitted, this is a natural consideration.
Our results apply to the most general semi-streaming settings where the entries of the input distance matrix, of size \(n^2\), arrive one by one in some arbitrary order, possibly adversarially.  Notably, all our algorithms require either one or two passes over the data while achieving constant factor approximations in polynomial time. Before discussing the key contributions of this work, we provide a formal definition of the problem.

\noindent{\bf Problem:} $\ell_p$ Best-Fit Ultrametrics/Tree-Metrics \smallskip\\
\noindent{\bf Input:} A set $V$ 
and a distance matrix \( D: \binom{V}{2} \rightarrow \mathbb{R}_{>0} \).
\smallskip\\
\noindent{\bf Desired Output:} An ultrametric (resp. tree metric) $T$ that spans $V$ and 
fits $D$ in the sense of minimizing:
\begin{equation*}
\|T-D\|_p=\sqrt[p]{\sum_{uv\in \binom{V}{2}} |T(uv)-D(uv)|^p}
\end{equation*}

\vspace{-2mm}
For \(p=0\), the aim is to minimize the total number of errors. In other words, each pair comes with a request regarding their distance in the output tree, and our goal is to construct a tree that satisfies as many of these requests as possible, minimizing the total number of pairs whose distances are altered. This fundamental problem for ultrametrics, also known as the \emph{Ultrametric violation distance} problem, was first investigated in~\cite{cohen2022fitting}, in which a novel constant-factor approximation algorithm in the RAM model was developed.
Charikar and Gao~\cite{CharikarG24} further improved the approximation guarantee to 5. 

We present for this problem a single-pass algorithm, in the semi-streaming setting, that provides a constant approximation and succeeds with high probability.
We remark that straightforwardly adapting this algorithm in the RAM model yields a near-linear time algorithm ($\widetilde{O}(n^2)$, while the input size is $\Theta(n^2)$), improving over the best known $\Omega(n^4)$ time from \cite{cohen2022fitting}\footnote{In \cite{cohen2022fitting} the exact running time of the algorithm has not been analyzed, but there exist inputs where it must perform $\Omega(n^2)$ repetitions of a flat-clustering algorithm that takes $\Omega(n^2)$ time per repetition.}.

\begin{restatable}{theorem}{lzerofit}\label{theorem:lzerofit}
There exists a single pass polynomial time semi-streaming algorithm that
w.h.p. $O(1)$-approximates the $\ell_0$ Best-Fit Ultrametrics problem.
\end{restatable}

Following, we show that this result also implies an approximation for the $\ell_1$ objective.

\begin{corollary} \label{cor:l1ApxUltra}
Let $\delta$ (resp. $\Delta$) be the smallest (resp. largest) absolute difference between distinct distances in $D$, for an $\ell_1$ Best-Fit Ultrametrics instance. There exists a single pass polynomial time semi-streaming algorithm that w.h.p $O(\Delta/\delta)$-approximates the $\ell_1$ Best-Fit Ultrametrics problem.
\end{corollary}
\begin{proof}
Let $T_0$ (resp. $T_1$) be an ultrametric minimizing $\|T_0-D\|_0$ (resp. $\|T_1-D\|_1$), and $T$ be the output from Theorem~\ref{theorem:lzerofit} (thus $\|T-D\|_0 = O(\|T_0-D\|_0)$.

It holds that $\|T_1-D\|_1 \ge \|T_1-D\|_0 \delta$, as in the $\ell_1$ objective we pay at least $\delta$ for each pair having a disagreement.
Similarly $\|T-D\|_1 \le \|T-D\|_0 \Delta = O(\|T_0-D\|_0 \Delta) = O(\|T_1-D\|_0 \Delta)$, by definition of $T_0$.
\end{proof}

It is interesting to note that for the $\ell_1$ objective, most recent approximation algorithms in the offline setting are not combinatorial, making it a significant challenge to adapt them to the semi-streaming model. 
The best known combinatorial approximation for \(\ell_1\) Best-Fit Ultrametrics and Tree-Metrics is \(O(\Delta/\delta)\), when \(\Delta/\delta = O(n)\) \cite{charikar, mcgregor}, achieved through the so-called pivoting algorithm
\footnote{The authors of \cite{charikar} claim the approximation is proportional to the number of distinct distances. That is because of a simplification they make in the paper, that the distances are all consecutive positive integers starting from $1$. For an example showing the $O(\Delta/\delta)$ analysis is tight, take $V=\set{u_1,u_2,u_3}$ and $D(uv)>\Delta, D(uw) = D(uv)-\delta, D(vw)=D(uv)-\Delta$. With probability $1/3$ we pick $u$ as the pivot, and pay $\Delta$, while the optimum solution only pays $\delta$.}.
Unfortunately, this algorithm is very challenging to adapt to a single-pass semi-streaming setting as it generalizes the PIVOT-based algorithm of Correlation Clustering, which, despite extensive research, has not been adapted to semi-streaming settings with only a single pass without significant modifications~\cite{BehnezhadCMT22, BehnezhadCMT23, CambusKLPU24, chakrabarty2023single}. 
Surprisingly, our \(O(\Delta/\delta)\) approximation for the \(\ell_1\) objective is derived directly from the algorithm for the \(\ell_0\) objective, eliminating the need to explicitly use this pivoting approach.

Further, we contrast Theorem~\ref{theorem:lzerofit} by ruling out the possibility of a single-round exact algorithm, even with sub-quadratic space and exponential time. For this, we provide a new lower bound result for the correlation clustering problem, showing that any single-pass streaming algorithm with sub-quadratic space cannot output the optimal clustering nor can maintain its cost.

\begin{restatable}{theorem}{lbzeroclusteringorscore}
    Any randomized single pass streaming algorithm that with probability greater than $\frac{2}{3}$ either solves the correlation clustering problem or maintains the cost of an optimal correlation clustering solution  requires $\Omega(n^2)$ bits.
\end{restatable}

We then extend this result to $\ell_p$ Best-Fit Ultrametrics problems for $p \in \{ 0,1 \}$, using the fact that correlation clustering is a special case of these problems (see e.g. \cite{charikar}).

\begin{restatable}{corollary}{lbzeroultrametric}
    For $p \in \{0,1\}$, any randomized single pass streaming algorithm that with probability greater than $\frac{2}{3}$ either solves $\ell_p$ Best-Fit Ultrametrics or just outputs the error of an optimal ultrametric solution requires $\Omega(n^2)$ bits.
\end{restatable}

Next, we consider the $\ell_\infty$ objective, where the goal is to minimize the maximum error.
In \Cref{section:linf} we provide a complete characterization of $\ell_\infty$ Best-Fit Ultrametrics in the semi-streaming model.
We give a single pass algorithm with $2$-approximation factor to this problem.

\begin{restatable}{theorem}
{inftysinglepass}\label{theorem:theorem_single_pass_md}
There exists a single pass polynomial time semi-streaming algorithm that $2$-approximates the $\ell_\infty$ Best-Fit Ultrametrics problem.
\end{restatable}

We contrast Theorem~\ref{theorem:theorem_single_pass_md} by showing that this is the best approximation factor achievable using a single pass, even with sub-quadratic space and exponential time.

\begin{restatable}{theorem}{linftylowerbound}
\label{theorem:l_infty_lower_bound}
Any randomized one-pass streaming algorithm for \(\ell_\infty\) Best-Fit Ultrametrics with an approximation factor strictly less than 2 and a success probability greater than \(\frac{2}{3}\) requires \(\Omega(n^2)\) bits of space.
\end{restatable}

Moreover, we demonstrate that allowing two passes is sufficient for an exact solution.
Therefore, we provide optimal tradeoffs between the number of passes and the approximation factor in all scenarios.

\begin{restatable}{theorem}{twopass}\label{theorem:twopassexact}
There exists a two-pass polynomial time semi-streaming algorithm that
computes an exact solution to the $\ell_\infty$ Best-Fit Ultrametrics problem.
\end{restatable}

In \Cref{section:trees} we show that all aforementioned algorithms can be extended to tree metrics.
This is achieved by providing reductions to the corresponding ultrametrics problems, requiring only one additional pass over the stream.
The reductions used for the $\ell_0$ and $\ell_\infty$ objectives differ significantly from each other. 

\begin{restatable}{theorem}{zeroreduction}\label{theorem:best_tree_zero}
There exists a two-pass polynomial time semi-streaming algorithm that w.h.p $O(1)$-approximates the $\ell_0$ Best-Fit Tree-Metrics problem.
\end{restatable}

Using the same arguments as in \Cref{cor:l1ApxUltra}, we obtain an analogous result for the $\ell_1$ objective.

\begin{corollary} Let $\delta$ (resp. $\Delta$) be the smallest (resp. largest) absolute difference between distinct distances in $D$, for an $\ell_1$ Best-Fit Tree-Metrics instance. There exists a two-pass pass polynomial time semi-streaming algorithm that w.h.p $O(\Delta/\delta)$-approximates $\ell_1$ Best-Fit Tree-Metrics problem.
\end{corollary}

\begin{restatable}{theorem}{inftyreduction}\label{theorem:best_tree_infty}
There exists a two-pass polynomial time semi-streaming algorithm that 6-approximates the $\ell_\infty$ Best-Fit Tree-Metrics problem.
\end{restatable}

\subsection{Technique Overview}
We provide a technical overview of the most technically novel contribution of our work, namely the results regarding the $\ell_0$ Best-Fit Ultrametrics algorithm in the semi-streaming model (more details in Section~\ref{section:l0ultra}).

\subsubsection{\texorpdfstring{Why previous $\ell_0$ approaches cannot be adapted}{Why previous l-0 approaches cannot be adapted}} \label{sec:othersDontWork}
In general, it is difficult to ensure a hierarchical structure while providing non-trivial approximation guarantees.
In Hierarchical Clustering research, such results usually rely on one of two standard approaches, namely the top-down (divisive) approach, and the bottom-up (agglomerative) approach.
In fact, with the exception of \cite{debarati}, results for $\ell_p$ Best-Fit Ultrametrics ($p<\infty$) \cite{mcgregor, charikar, debarati, cohen2022fitting, CharikarG24} all rely on the divisive approach.

\vspace{-3mm}

\paragraph{Non-divisive approaches.} The only relevant result applying a non-divisive approach is that of \cite{debarati}, which crucially relies on a large LP.
Unfortunately, it is not known how to solve such an LP in streaming.

\vspace{-3mm}

\paragraph{Divisive approaches.} A divisive algorithm starts with the root node (containing the whole $V$), computes its children (subsets $V'$ at height $h$) based on some division strategy, and recurses on its children. Different division strategies have been employed, with the most prominent ones using the solution to an LP, or attempting to satisfy a particular (usually randomly chosen) element called the pivot, or solving some flat clustering problem. In what follows, we discuss why existing division strategies do not work in our case.

\subparagraph{Correlation Clustering.} Perhaps the most straightforward approach is to solve a (flat) clustering problem; for each pair of vertices $u,v\in V'$, we ideally want them together if $h>D(uv)$, and apart otherwise.
This corresponds to the Correlation Clustering problem, which returns a clustering violating as few of our preferences as possible.
Unfortunately, this approach does not work for $\ell_0$ Best-Fit Ultrametrics, as certain choices that appear good locally (on a particular height $h$) may be catastrophic globally.

\subparagraph{The first result for $\ell_0$ Best-Fit Ultrametrics.} The authors of \cite{cohen2022fitting} overcame the shortcomings of Correlation Clustering as a division strategy by solving a particular flavor of it, called Agreement Correlation Clustering. This guaranteed further structural properties\footnote{Informally, when we solve Agreement Correlation Clustering we obtain a clustering $\mathcal{C}$ with all its clusters being dense and the property that there exists a near-optimal clustering $\mathcal{C'}$ such that every cluster of $\mathcal{C'}$ is a subset of some cluster in $\mathcal{C}$.} that could be leveraged to provide $O(1)$ approximation for $\ell_0$ Best-Fit Ultrametrics. However this approach is too strong to guarrantee in streaming, since one can recover adjacency information with black-box calls to an Agreement Correlation Clustering subroutine. This of course requires $\Theta(n^2)$ bits of memory, while in streaming we only have $\widetilde{O}(n)$.

\subparagraph{Other results for $\ell_0$ Best-Fit Ultrametrics.} The other results for $\ell_0$ Best-Fit Ultrametrics are pivot-based and do not work in our case. Indeed, one of them \cite{CharikarG24} is based on a large LP for which no streaming solution is known, while the other one \cite{cohen2022fitting} is combinatorial but with approximation factor $\Omega(\log{n})$.

\subsubsection{Our techniques}\label{section:techniques}

\paragraph*{$\ell_0$ Best-Fit Ultrametrics.}
Our streaming algorithm is a divisive algorithm.
In the divisive framework, each level of the tree is defined by a distinct distance from the input, which allows each level to be visualized as an instance of the correlation clustering problem.
In this instance, two vertices are connected if their distance is at most the threshold associated with that level; otherwise, they are not connected.
Following this, different layers of the ultrametric tree are built by repeatedly applying a division strategy in a top-down fashion.
Here, we highlight the techniques we develop to design a semi-streaming algorithm that uses only a single pass and computes an $O(1)$ approximation for the $\ell_0$ Best-Fit Ultrametrics problem.

\subparagraph*{Distances Summary.} 

The first fundamental challenge is identifying which distances should be preserved in the constructed ultrametric, given that the input may contain $\Omega(n^2)$ distinct distances.
A divisive algorithm may need to perform its division strategy on every level defined by such a distance (of course, sometimes it may decide not to divide anything); however, in the semi-streaming setting, we cannot even afford to store all these distances. 
Instead, we work with a compressed set of distances that effectively captures all important information. More formally, we focus on distances \(d\) for which there exists at least one vertex \(u\) such that the number of vertices with distance less than \(d\) from \(u\) is significantly smaller than the number of vertices with distance at most \(d\) from \(u\). Using this notion, we demonstrate that it is sufficient to consider only a near-linear number of distances to achieve a good approximation.

\subparagraph*{Agreement Sketches.} A key component in many (flat) clustering algorithms 
(including the first algorithm for Correlation~Clustering~\cite{DBLP:conf/focs/BansalBC02}, which inspired many others, such as \cite{cohen2021correlation, Assadi022, cohen2022fitting, abs-2404-06797}) 
involves comparing the set of neighbors of two vertices. While our division strategy also builds on such comparisons, both the hierarchical nature and streaming constraints of our setting present unique challenges.
In Correlation Clustering each vertex has only a single set of neighbors, however, in our hierarchical setting, each layer of the tree is associated with a different distance threshold, producing different sets of neighbors for a vertex. In the worst case, there can be $O(n)$ such sets, and building a sketch for each can require up to quadratic space. 

To address this,
we build a new sketch for a node only when its set of neighbors changes significantly.
The intuition here is that if the neighborhood of a node has not changed substantially, then a precomputed sketch for a nearby neighborhood will suffice. However, implementing this in a semi-streaming
setting, where distances between pairs can arrive in any arbitrary order, is challenging. Since we
cannot store the distances to all other nodes from a given node simultaneously, identifying significant
changes in a node’s neighborhood becomes difficult. To manage this, we develop a new technique that combines random sampling with a pruning strategy, ensuring that the overall space required to store all the sketches is \(\tilde{O}(n)\).

In this approach, we build each sketch by randomly sampling nodes. Assuming the neighborhood size has dropped substantially, we expect the correct sketch to reach a certain size. Notably, the set of neighbors of a node only shrinks as the distance decreases. Thus, for a specific weight threshold, if the sample (or sketch) size grows considerably, this indicates that the neighborhood has not changed much, so we disregard that weight threshold and delete the corresponding sample from the sketch. Specifically, for each node, we build and store sketches when the neighborhood size shrinks by a constant factor. Following this, we consider at most \(\log n\) different sizes, and storing the sketch for all sizes takes only polylogarithmic space for a node. Therefore, the total space required to store the sketches for all the nodes is bounded by \(\tilde{O}(n)\). Moreover, the sketches ensure that for each node \(u\) and each weight \(w\), there exists a weight \(w'\) such that the neighboring nodes of \(u\) at \(w\) and \(w'\) differ very little, and we have built a sketch corresponding to the neighboring set at weight \(w'\).

\subparagraph*{Across-Levels Correlations.} In divisive algorithms, while building a new level of the ultrametric tree, the recursions performed depend on the divisions computed at previous recursion levels.
In this sense, the divisive framework can be viewed as an adaptive adversary for the division strategy we need to perform.
This is not an issue when deterministic division strategies are used (e.g. as in \cite{cohen2022fitting}), but it becomes particularly problematic in our case, where we are forced to use random sketches because of the issue with the $\Omega(n^2)$ distinct distances.

The challenge here arises from the fact that for a given vertex we do not build a new sketch for each level of the tree. Instead, we only construct a sketch when the set of neighbors changes substantially. Consequently, multiple levels of the tree must reuse the same sketch, which increases the correlation among clusters at different levels. This makes it difficult to ensure concentration bounds when limiting the overall error.

To address this, our approach aims to ``limit'' the dependencies by ensuring our algorithm has only a logarithmic recursion depth (as opposed to the \(\Omega(n^2)\) recursion depth in straightforward divisive approaches). This allows us to afford independent randomness by using a new sketch for each recursion depth.
To reduce the recursion depth, we make the following observation: if the correlation clustering subroutine identifies a large cluster (e.g., containing a $0.99$ fraction of the vertices), we can detect this cluster without explicitly applying the correlation clustering algorithm (thus omitting the requirement of using a sketch). This is because all vertices within this cluster have large degrees, while those outside have very small degrees. Therefore, it suffices to identify the vertices with small degrees and remove them to generate the new cluster.
It is important to note that the degree calculation must consider the entire graph, not just the subgraph induced by the current cluster being considered. Otherwise, intra-recursion dependencies could be introduced, and thus the logarithmic recursion depth guarantee may not suffice.

\subparagraph*{Within-Level Correlations.} Correlation issues do not only occur vertically (across levels), but also horizontally (within the same level), as most algorithms for correlation clustering compute a cluster, and then recurse on the rest of the elements.
However in our case, such an adaptive construction may lead to the possibility of reusing sketches, making it difficult to ensure concentration.

We overcome these issues in several ways. 
First, we use these sketches to compute the agreement among vertices (i.e., computing the similarity between the set of neighbors of each pair of vertices) before we start constructing any clusters.
Finally we propose an algorithm that is relying solely on our agreement sketches and is decomposed into independent events, thus only requiring us to consult the sketches of each layer only once. By using the agreements precomputation and our proposed clustering algorithm we ensure concentration while limiting the error for all clusters within a level.

\subparagraph*{$S$-Structural Clustering.} Finally, as argued in \Cref{sec:othersDontWork}, we need a division strategy that is different from the existing Agreement Correlation Clustering of~\cite{cohen2022fitting}.
That is because it can be proven that solving Agreement Correlation Clustering on arbitrary subgraphs requires $\Omega(n^2)$ bits of memory.

Instead, we introduce $S$-Structural Clustering, which is inspired by Agreement Correlation Clustering. The key distinction is that now we require a clustering of $S$ to satisfy the structural properties, while also considering edges with only one endpoint in $S$.
This distinction is exactly what allows us to generalize our proposed algorithm to solve $S$-Structural Clustering by relying solely on the global neighborhoods of its vertices. 
Interestingly, the resulting time complexity of our general algorithm only depends on the size of the subgraph, as we compress all the necessary global information through our sketches.
Finally, we remark that both the construction of the sketches (Section~\ref{section:sketches}) and the introduction of the \( S \)-Structural Clustering (Section~\ref{sec:Sstructure}) are two novel contributions of our work and could be of independent interest.

\paragraph*{$\ell_0$ Best-Fit Tree-Metrics.}
In \cite{kipouridis2023fitting} it is shown how to reduce $\ell_0$ Best-Fit Tree-Metrics to $\ell_0$ Best-Fit Ultrametrics.
In this approach, however, one needs to create $n$ different instances of $\ell_0$ Best-Fit Ultrametrics, which is not feasible in the semi-streaming model.
In this work, we show that randomly solving a logarithmic number of these $n$ different instances suffices.

Our initial approach requires $3$ passes over the stream. One for a preprocessing step implicitly constructing the $\ell_0$ Best-Fit Ultrametrics instances, one to solve these instances (and post-process them to extract trees that solve the original problem), and a final one to figure which one of the logarithmically many trees we need to output (the one with the smallest cost is picked).

We further improve the number of passes to $2$, by eliminating the need for the final pass.
To do that, we note that there are many trees with ``tiny'' cost related to the input; let $A$ be the set containing these trees.
By triangle inequality, all trees in $A$ have ``small'' cost related to each other.
If we create a graph with the trees as nodes, and an edge between two nodes when the cost relative to each other is small, we then show that with high probability this graph contains a big clique.
Finally, we show that any node (corresponding to a tree) from a big clique is a good enough approximation to the original input, even if it is not in $A$.

\paragraph*{$\ell_\infty$ Results.}
Regarding $\ell_\infty$ Best-Fit Ultrametrics, we show that the existing exact algorithm \cite{farach1993robust} can be straightforwardly adapted to a $2$-pass semi-streaming algorithm.
Naturally, as the problem has been solved exactly, no research has focused on approximation algorithms.
In this work we show that the solution to a related problem ($\ell_\infty$ Min-Decrement Ultrametrics) $2$-approximates $\ell_\infty$ Best-Fit Ultrametrics.
Then we adapt the exact solution for $\ell_\infty$ Min-Decrement Ultrametrics \cite{farach1993robust} to obtain a single pass semi-streaming algorithm.

We also show that no single-pass semi-streaming algorithm can give a better-than-2 approximation, for otherwise we could compress any graph in $\widetilde{O}(n)$ space.
Together, these results completely characterize $\ell_\infty$ Best-Fit Ultrametrics in the semi-streaming setting, regarding the optimal number of passes and the optimal approximation factor.

For $\ell_\infty$ Best-Fit Tree-Metrics, there exists a reduction to Ultrametrics \cite{agarwala}, blowing up the approximation by a factor $3$.
Adapting it in the semi-streaming requires one additional pass through the stream.
Using it with our $2$-approximation for $\ell_\infty$ Best-Fit Ultrametrics (rather than with the exact algorithm, as done in \cite{agarwala}),
we need $2$ passes (instead of $3$).

\newcommand{\clusters}{\mathcal{C}}
\newcommand{\inputmatrix}{D}

\section{Preliminaries}
We start by presenting useful notations we employ throughout the text.
We use $uv$ to denote an unordered pair $\set{u,v}$.
We use the term distance matrix to refer to a function from $\binom{V}{2}$ to the non-negative reals.
Let $D$ be a distance matrix.
For easiness of notation, we use $w_{max} = \max_{uv}{D(uv)}$.
We slightly abuse notation and say that for any $u\in V$, $D(uu)=0$.
For $p \ge 1$, $\|D\|_p = \sqrt[p]{\sum_{uv \in \binom{V}{2}}|D(uv)|^p}$ is the $\ell_p$ norm of $D$.
We extend the notation for $p=0$.
In this case, $\|D\|_0$ denotes the number of pairs $uv$ such that $D(uv)\ne 0$.
We even say $\|D\|_0$ is the $\ell_0$ norm of $D$, despite $\ell_0$ not being a norm.

If $T$ is a tree and $u,v$ are two nodes in $T$, then we write $T(uv)$ to denote the distance between $u$ and $v$ in $T$.
An ultrametric is a metric $(V,D)$ with the property that $D(uv) \le \max\set{D(uw),D(vw)}$ for all $u,v,w\in V$.
It holds that $(V,D)$ is an ultrametric iff there exists a rooted tree $T$ spanning $V$ such that all elements of $V$ are in the leaves of $T$, the depth of all leaves is the same, and $D(uv)=T(uv)$ for all $u,v\in V$.
We call trees with these properties \emph{ultrametric trees}.

In the semi-streaming model, the input is again a distance matrix $D$ on a vertex set $V$.
Let $n=|V|$.
Our algorithm has $\widetilde{O}(n)$ available space, and the entries of $D$ arrive one-by-one, in an arbitrary order, as pairs of the form $(uv,D(uv))$.
For simplicity, we use the standard assumption that each distance $D(uv)$ fits in $O(\log{n})$ bits of memory.

We let $E_w$ be the set of pairs $uv$ such that $D(uv) \le w$. 
We define $N_w(u)$, the set of neighbors of $u$ at level $w$, to be the vertices $v$ such that $uv \in E_w$ (including $u$ itself), and the degree of $u$ at level $w$ to be $d_w(u) = |N_w(u)|$.
We even write $N(u)$ and $d(u)$ (instead of $N_w(u)$ and $d_w(u)$) when $E_w$ is clear from the context.
Given an ultrametric tree $T$, a cluster at level $w$ is a maximal set of leaves such that every pairwise distance in $T$ is at most $w$.
It is straightforward that a cluster at level $w$ corresponds to the set of leaves descending from a node of $T$ at height $w/2$.
Abusing notation, and only when it is clear from the context, we refer to this node as a cluster at level $w$ as well.

Regarding $\ell_0$, it is sufficient to focus on ultrametrics where the distances between nodes are also entries in $D$.
That is because if an ultrametric $T$ does not have this property, we can create an ultrametric $T'$ with this property such that $\|T'-D\|_0 \le \|T-D\|_0$ (folklore).
To do this, simply modify every distance $d$ in $T$ to the smallest entry in $D$ that is at least as large as $d$ (if no such entry in $D$ exists, then we modify $d$ to be the maximum entry in $D$).

\newcommand{\sizes}{\mathbb{S}}
\newcommand{\counter}{counter}

\section{\texorpdfstring{$\ell_0$ Ultrametrics}{l-0 Ultrametrics}} \label{section:l0ultra}
In this section, we show how to $O(1)$-approximate $\ell_0$ Best-Fit Ultrametrics with a single pass in the semi-streaming model. Formally we show the following.

\lzerofit*

Our algorithm consists of two main phases. In the streaming phase, we construct efficient sketches 
that capture the essential information of the input matrix $D$. That is, we store a compressed representation of $D$, denoted as $\widetilde{D}$, which, unlike $D$, has a reduced size of $\widetilde{O}(n)$ rather than $O(n^2)$ values (hereafter called \emph{weights}). 
Yet, we will show in Section~\ref{section:sketches} that for every weight $w \in D$ and every $u \in V$, a weight $\tilde{w} \in \widetilde{D}$ is stored, such that $N_w(u)$ and $N_{\tilde{w}}(u)$ are roughly the same. This guarantee enables us to approximate both the size of a neighborhood and the size of the intersection for two different neighborhoods.

The second step is a post-stream process that carefully utilizes the precomputed sketches while addressing the adaptivity challenges discussed in \Cref{section:techniques}.
In \Cref{section:structuralClustering} we show how to compute the $S$-Structural Clustering subroutine, which we will use as our division strategy.
In \Cref{section:algorithm} we present our main algorithm, which uses this subroutine and the distances summary as black-boxes to construct the ultrametric tree. 
Finally, in \Cref{section:LowerBounds}, we establish the necessity of approximation in the streaming setting by proving that computing an optimal solution requires $\Omega(n^2)$ bits of memory.

\subsection{Construction of Sketches}\label{section:sketches}
This section outlines the process for constructing sketches that enable our algorithm's implementation. For now we consider large neighborhoods of size $\Omega (\log^4 n)$. 
While a similar approach was used in~\cite{cohen2021correlation}\footnote{In ~\cite{cohen2021correlation}, the authors claim polylogarithmic size sketches for each vertex. However, we are unable to verify this. Specifically, the random set is constructed by selecting each vertex with probability \(\min\left\{\frac{a \log n}{\beta j}, 1\right\}\), where \(a\) is a constant. Since \(j\) is at most \(O\left(\frac{\log n}{\beta}\right)\), the probability of selecting a vertex is at least \(\min\{\Omega(a), 1\}\), which is a constant. Thus, each random set is of size \(\Omega(n)\). Therefore, for a vertex \(v\) with \(|N(v)|=\Omega(n)\), the sketch size will be of size \(\Omega(n)\).}, for the problem of correlation clustering, the challenge here is different. Unlike correlation clustering, where each vertex has only a single set of neighbors, each layer of the tree in our context is associated with a different distance threshold. Thus each varying threshold can produce a different set of neighbors for a vertex. In the worst case, there can be \(n\) such sets, and building a sketch for each changing set of neighbors for each vertex can require up to quadratic space (or even cubic, if implemented naively).

We denote the weight of an edge $e=uv$ by $w(e)=D(uv)$.
Each sketch is constructed for a specific vertex with a predetermined size chosen from the set $\sizes = \{n, \frac{n}{(1+\zeta)}, \frac{n}{(1+\zeta)^2}, \dots, \log^4 n\}$, where $\zeta$ is a small constant parameter to be adjusted.
Each sketch will encapsulate a neighborhood of the vertex of size roughly $s$, and allow us to compare the common intersection of two different neighborhoods. 
Let \(w^v_s\) be the largest weight for which \(\frac{s}{1+\zeta}<|N_{w^v_s}(v)| \le s\). We call size \(s\) \emph{relevant} for vertex $v$ if such \(w^v_s\) exists.

To obtain the sketches, for each $s' \in \sizes$, we start by generating a random subset \(R_{s'} \subseteq [n]\) by sampling each vertex from $V$ independently with probability $\log^2 n/s'$, prior to processing the stream. 
For each vertex $v$, each relevant size $s$, and each $s'$ satisfying $\frac{1}{2}s \leq s' \leq s$, we define a sketch \(\Sk^{v}_{s,s'}\). 
Every sketch consists of (i) an estimate of the parameter \(w^v_s\), denoted by \(\tilde{w}^v_s\), and (ii) an (almost) random sample of \(v\times N_{\tilde{w}^v_s}(v) \) of size \(O(\log^2 n)\), along with the weight of each sampled edge. To achieve this, we store a collection of edges \(C_1^v, \dots, C_\ell^v \subseteq v\times N_{\tilde{w}^v_s}(v)\), where all edges in \(C_i^v\) have the same weight $w^v_i$, which we also store alongside \(C_i^v\), and let $C_\ell^v$ be the collection corresponding to the largest weight. Furthermore, we ensure that the overall size of all collections \(\sum_{i \in [\ell]} |C_i^v|\) is  \(O(\log^3 n)\) bits.

The purpose of incorporating two size parameters, $s$ and $s'$ into the sketch is to enable comparisons of neighborhoods that have slightly different, but relatively close, sizes. Yet, for simplicity, the reader may assume $s=s'$ for the following construction and claims. 
We now describe the process of constructing a specific sketch $\Sk^{v}_{s,s'}$ given the input stream: 
\begin{enumerate}
    \item Initialize $\counter=0$, \(w_{m} = 0\).
    \item If \(e\) is not incident on \(v\), continue to the next edge. 
    \item Else, if \(w_{m} \neq 0\) and \(w(e) \ge w_{m}\), continue to the next edge. 

    \item Else if $u\notin R_{s'}$, where $e=(u,v)$, continue to the next edge. 

    \item Otherwise proceed as follows:
    \begin{enumerate}
        \item If there is a collection of edges \(C_i^v\) with an associated weight of \(w(e)\), add the edge \(e\) to \(C_i^v\). Otherwise, create a new collection \(C_i^v\) containing the edge \(e\) alongside \(w(e)\).

        \item Increase $\counter$ by $1$.

        \item If $\counter > (1+\frac{\zeta}{2}) \frac{s}{s'} \log^2 n$, delete $C_\ell^v$ the collection with the largest weight associated with it, set $w_{m}=w^v_\ell$ and $\counter = \counter - \abs{C_\ell^v}$.
    \end{enumerate}
\end{enumerate}

After processing all the edges, output \({\Sk}^{v}_{s,s'} = \bigcup_{i \in [\ell]} (C_i^v \times \{w^v_i\})\), furthermore if $s=s'$ we let \(\tilde{w}^v_s = \max_{i \in [\ell]} w^v_i\) and call it a $\emph{governing weight}$ of the sketches parametrized by $v$ and $s$. Namely, $\tilde{w}^v_s$ is the weight associated with the neighborhood a sketch parametrized by $v$ and $s$ is encapsulating. The next claim shows that this sketches can be stored in semi-streaming settings.

\begin{claim}
\label{clm:sketch2size}
The sketches \({\Sk}^{v}_{s,s'}\), where  \(v \in V\) and \(s \in \sizes\), can be constructed and stored in \(O(n \log^4 n)\) bits. 
\end{claim}

\begin{proof}
First, note that the total space required to store all \( R_{s'} \), where \( s' \in \sizes \), is \( O(n\log^3 n) \) bits. Next, each sketch \({\Sk}^{v}_{s,s'}\) stores at most \( O(\log^2 n) \) edges, thus requires \( O(\log^3 n) \) bits. Since we build \( O(\log n) \) different sketches for each vertex, the overall space required is \( O(n\log^4 n) \).
\end{proof}

The next claim demonstrates that for every $w^v_s$ there is a sketch with governing weight $\tilde{w}^v_s$ such that $\abs{N_{\tilde{w}^v_s}(v)}$ is a good approximation to $\abs{N_{w^v_s}(v)}$.

\begin{claim}
\label{clm:sketch1nb}
With high probability, for each vertex \(v\) and each relevant size \(s \in \sizes \), we have \( (1-\zeta)\abs{N_{w^v_s}(v)} \leq \abs{N_{\tilde{w}^v_s}(v)} \leq (1+\zeta)^2\abs{N_{w^v_s}(v)} \).
\end{claim}

\begin{proof}
First, we consider the case where \(\tilde{w}^v_s \leq w^v_s\), that is, \(N_{\tilde{w}^v_s}(v) \subseteq N_{w^v_s}(v)\). We prove that with high probability, \(|N_{w^v_s}(v) \setminus N_{\tilde{w}^v_s}(v)| \leq \zeta |N_{w^v_s}(v)|\). 

Otherwise, if \(|N_{w^v_s}(v) \setminus N_{\tilde{w}^v_s}(v)| > \zeta |N_{w^v_s}(v)|\), we claim that at least one of the following two bad events must occur. 
We define the first bad event as \(B_1\), where no edge is sampled from \((v\times N_{w^v_s}(v)) \setminus (v\times N_{\tilde{w}^v_s}(v))\). We define the second bad event as \(B_2\), where more than \((1+\frac{\zeta}{2}) \log^2 n\) edges are sampled from \(v\times N_{w^v_s}(v)\). If neither of these bad events occurs, then at least one edge \(e\) is sampled from \((v\times N_{w^v_s}(v) )\setminus (v\times N_{\tilde{w}^v_s}(v))\), where \(\tilde{w}^v_s < w(e) \le w^v_s\), and the associated collection of $w(e)$ is not deleted. Consequently, \(e\) should survive, contradicting the claim that \(\tilde{w}^v_s\) is the maximum weight of an edge that is sampled and not deleted. Since $s$ is relevant and $\frac{s}{1+\zeta} \leq \abs{N_{w^v_s}(v)} \leq s$, both events $B_1$ and $B_2$ occurs with probability at most \(1/n^{10}\) using Chernoff bound.

Next, we consider the case where \(\tilde{w}^v_s > w^v_s\), and thus \(N_{w^v_s}(v) \subset N_{\tilde{w}^v_s}(v)\). We prove that with high probability, \(|N_{\tilde{w}^v_s}(v)| \leq (1+\zeta)^2|N_{w^v_s}(v)|\).

Otherwise, if \(\abs{N_{\tilde{w}^v_s}(v)} > (1+\zeta)^2 \abs{N_{w^v_s}(v)}\), we claim that the following bad event \(B\) must occur. We define \(B\) as the event where at most \((1+\frac{\zeta}{2}) \log^2 n\) edges are sampled from \(v\times N_{\tilde{w}^v_s}(v)\). If more than \( (1+\frac{\zeta}{2}) \log^2 n\) edges are sampled from \(v\times N_{\tilde{w}^v_s}(v)\), then $\tilde{w}^v_s$ cannot be obtained. According to Chernoff bound, since \(\abs{N_{\tilde{w}^v_s}(v)} > (1+\zeta)^2 \abs{N_{w^v_s}(v)} > (1+\zeta)s\), \(B\) occurs with probability at most \(1/n^{10}\). Therefore, the probability that \(\abs{N_{\tilde{w}^v_s}(v)} > (1+\zeta)^2 \abs{N_{w^v_s}(v)}\) is at most \(1/n^{10}\).

As there are $n$ different choices for $v$, and $O(\log n)$ choices for $s$, the claim holds for all $w^v_s$ w.h.p. 
\end{proof}

We now extend this result for every weight $w$, and show how to obtain a sketch that is a good approximation to $N_w(v)$.

\begin{claim}\label{clm:sketch1nb2}
For each vertex \(v\) and each weight $w$ with $\abs{N_w(v)} \geq \log^4 n$, we can report a sketch associated with size $s$ and governing weight $\tilde{w}^v_s$, such that with high probability, 
$\frac{\abs{N_{\tilde{w}^v_{s}}(v)}}{1+5\zeta} \leq |N_{w}(v)| \leq \frac{|N_{\tilde{w}^v_s}(v)|}{1-\zeta}$
\end{claim}

\begin{proof}
Let $\tilde{w}^v_{+}$ (resp. $\tilde{w}^v_{-}$) be the immediate governing weights above (resp. below) $w$ within all the sketches of $v$. We count the number of sampled edges in the sketch associated with the weight $\tilde{w}^v_{+}$ of weight greater than $w$. If there are less than $4\zeta \log^2 n$ such edges then we report this sketch, and otherwise we report the sketch associated with the weight $\tilde{w}^v_{-}$.

If there are less than  $4\zeta \log^2 n$ edges with weight greater than $w$ in the sketch, then using Chernoff bound w.h.p we have $\abs{N_{\tilde{w}^v_{s}}(v)} \leq (1+5\zeta) \abs{N^v_{w}(v)}$.

However, if $\abs{N_{\tilde{w}^v_{+}}(v)} > (1+5\zeta) \abs{N^v_{w}(v)}$, that is, $\abs{N_{\tilde{w}^v_{+}}(v)} - \abs{N^v_{w}(v)} > 5 \zeta \abs{N^v_{w}(v)}$, then using Chernoff bound we deduce that w.h.p we report the sketch associated with $\tilde{w}^v_{-}$.
We now prove that this sketch is a good approximation. Let \(s \in \sizes \) be such that \(\frac{s}{1+\zeta}<|N^v_{w}(v)| \le s\). By definition, $s$ is relevant for $v$, and 
$\frac{|N_{w^v_s}(v)|}{1+\zeta}< |N^v_{w}(v)| \leq |N_{w^v_s}(v)|$. Following Claim~\ref{clm:sketch1nb}, there exists a sketch $(\tilde{w}^v_s, {\Sk}^{v}_s)$, such that w.h.p. \((1-\zeta)\abs{N_{w^v_s}(v)} \leq \abs{N_{\tilde{w}^v_s}(v)} \leq (1+\zeta)^2\abs{N_{w^v_s}(v)}\). Thus, w.h.p. there exist a sketch $\tilde{w}^v_s$ such that, \(\frac{|N_{\tilde{w}^v_s}(v)|}{(1+\zeta)^3} \leq |N_{w}(v)| \leq \frac{|N_{\tilde{w}^v_s}(v)|}{1-\zeta}\). Since $\frac{\abs{N_{\tilde{w}^v_{+}}(v)}}{1+5\zeta} >  \abs{N^v_{w}(v)}$, it must hold that $|N_{w}(v)| \leq \frac{|N_{\tilde{w}^v_{-}}(v)|}{1-\zeta}$. 
Overall, the reported governing weight satisfies w.h.p, $\frac{\abs{N_{\tilde{w}^v_{s}}(v)}}{1+5\zeta} \leq |N_{w}(v)| \leq \frac{|N_{\tilde{w}^v_s}(v)|}{1-\zeta}$.
\end{proof}

\vspace{3mm}
\noindent

We conclude this section by providing another data structure for storing the nearest \(2\log^4 n\) neighbors for each vertex $v$, denoted by \(N_{\close}(v)\). This will allow us to compare neighborhoods of small size.
The implementation of \(N_{\close}(v)\) is done using a priority queue with predefined and fixed size \(2\log^4 n\). We add every edge incident to $v$ to the priority queue \(N_{\close}(v)\) together with the associated edge. This leads to the following claim.

\begin{claim}
\label{clm:sketch3}
 For each vertex \(v\), \(N_{\close}(v)\) can be stored in \(O(\log^5 n)\) bits, and it contains the nearest \(2\log^4 n\) neighbors of \(v\).
\end{claim}

In this section, we have outlined the construction of sketches with a total space of $\tilde{O}(n)$, showing that for each vertex $v$ and weight $w$, we can report a sketch associated with  governing weight \(\tilde{w}^v_s \) that with high probability is a random sample of a neighborhood $N_{\tilde{w}^v_s}(v)$ that is roughly of the same size as $N_w(v)$. 
In the next section, we will demonstrate how these sketches can be utilized to estimate the size of the symmetric difference in a way that supports the algorithm’s requirements, justifying the need for maintaining several sketches for each choice of $v$ and $s$.

\subsection{Structural Clustering} \label{section:structuralClustering}
\newcommand{\agreez}[1]{A^3(#1)}

In this section, we introduce an algorithm that requires a single pass over the input stream to solve $S$-Structural Clustering. This extends the notion of Agreement Correlation Clustering from \cite{cohen2022fitting}, to which we refer as $V$-Structural Clustering. Our semi-streaming algorithm hinges on the key idea that clusters should be formed from vertices that share almost similar neighborhoods.
We emphasize that our algorithm is also applicable in the standard RAM (non-streaming) setting and runs in near-linear $\widetilde{O}(|S|^2)$ time, for $S\subseteq V$, improving the $\Omega(|V|^3)$ time algorithm previously known for $V$-Structural Clustering.

We begin our presentation in \Cref{section:structural_properties} with an algorithm solving $V$-Structural Clustering, designed to be adapted in the semi-streaming model.
Then, in \Cref{sec:Sstructure} we show how our proposed algorithm could also be extended to compute $S$-Structural Clustering, which is our division strategy for constructing ultrametrics in \Cref{section:algorithm}.
Finally, in \Cref{section:semi_agreement_cc} we show how to actually implement these algorithms in the semi-streaming model, by utilizing our sketches outlined in \Cref{section:sketches}.

\subsubsection{Algorithm for V-Structural Clustering (Agreement Correlation Clustering)}\label{section:structural_properties}

We begin by solving the Structural Clustering problem for the entire vertex set $V$.
Our graph is $(V,E=E_W)$ for some weight $W$. First we present the definitions of \textit{agreement} and \textit{heavy} vertices as in~\cite{cohen2021correlation}. 
The parameters $\beta$ and $\epsilon$ that appear in the following definitions are sufficiently small constants. Furthermore we denote by $\triangle$ the symmetric difference between two sets, that is, $A \triangle B = A \setminus B \cup B \setminus A$.

\begin{definition}[agreement]
\label{definition:agreement}
Two vertices $u, v$ are in $\beta$-agreement iff $\abs{N(u)\triangle N(v)} < \beta \max\{d(u), d(v)\}$, which means that $u, v$ share most of their neighbors. $A(u)$ is the set of vertices in $\beta$-agreement with $u$.
\end{definition}

\begin{definition}[heavy]
\label{definition:heavy}
We say that a vertex $u$ is $\epsilon$-heavy if $\abs{N(u) \setminus A(u)} < \epsilon d(u)$, which means that most of its neighbors are in agreement with $u$. Denote by $H(u)$ the $\epsilon$-heaviness indicator of vertex $u$.
\end{definition}

Computing the $\beta$-agreement set $A(u)$ of a vertex and its $\epsilon$-heaviness indicator $H(u)$ is a crucial part of the algorithm. Normally, both can be computed exactly by applying the definitions, even using a deterministic algorithm. However, in the semi-streaming model, we can only approximate $A(u)$ and $H(u)$ with high probability. 
In \Cref{section:semi_agreement_cc} we show that, using the sketches outlined in Section~\ref{section:sketches}, we can achieve a sufficient approximation that allows us to solve the Structural Clustering for $V$ with high probability.

We allow the following relaxations. Let $A(v)$ be a set containing all vertices that are in $0.8$ agreement with $v$ and no vertices that are not in $\beta$ agreement with $v$. Similarly let $\agreez{v}$ be a set containing all vertices that are in $2.4\beta$-agreement with $v$ and no vertices that are not in $3\beta$-agreement with $v$. And finally let $H(v)$ be a method that returns true if $v$ is $\epsilon$-heavy and false if $v$ is not $1.2\epsilon$-heavy.

With these tools and definitions at our disposal, we can introduce Algorithm~\ref{algorithm:simplified_structural_clustering}. It is important to note that this algorithm is executed, using the sketches alone, post stream process. Given the respective sketches, the time complexity of the algorithm is $\widetilde{O}(|V|^2)$.

\begin{algorithm}[H]
\caption{$V$-Structural-Clustering}
\label{algorithm:simplified_structural_clustering}
\begin{algorithmic}[1]
    \For{$v \in V$}
        \If{$H(v)$ \textbf{and} $v$ is not already included in an existing cluster}
            \State Create a new cluster $\agreez{v}$ \label{alg:cluster_creation}
        \EndIf
    \EndFor
    \State Create singleton clusters for all remaining vertices.
\end{algorithmic}
\end{algorithm}

We next show that Algorithm~\ref{algorithm:simplified_structural_clustering} is guaranteed to return a set of disjoint clusters that satisfy the required structural properties. We are referring to the special properties of $V$-Structural Clustering, which are expressed in terms of the definitions of \textit{important} and \textit{everywhere dense} groups as in~\cite{cohen2022fitting}.

\begin{definition}[important group]
\label{definition:important_group}
Given a correlation clustering instance, we say that
a group of vertices $C$ is important if for any vertex $u\in C$, $u$ is adjacent to at least $(1-\epsilon)$ fraction of the vertices in $C$ and has at most $\epsilon$ fraction of its neighbors outside of $C$.
\end{definition}

\begin{definition}[everywhere dense]
\label{definition:everywhere_dense}
Given a correlation clustering instance, we say that a group of vertices $C$ is everywhere dense if for any vertex $u\in C$, $u$ is adjacent to at least $\frac{2}{3}|C|$ vertices of $C$.
\end{definition}

Lemma~\ref{theorem:structural_properties} formally outlines the supplementary properties required for a correlation clustering algorithm to qualify as structural, demonstrated in the context of Algorithm~\ref{algorithm:simplified_structural_clustering}.

\begin{lemma} [structural properties]
\label{theorem:structural_properties}
    Suppose $\beta = 5\epsilon(1+\epsilon)$ for a small enough parameter $\epsilon \leq 1/95$. Let $\clusters$ be the set of clusters returned by Algorithm~\ref{algorithm:simplified_structural_clustering}. Then, for any important group of vertices $C' \subseteq V$, there is a cluster $C \in \clusters$ such that $C' \subseteq C$, and $C$ does not intersect any other important groups of vertices disjoint from $C'$. Moreover, every
    cluster $C \in \clusters$ is everywhere dense.
\end{lemma}

In order to prove Lemma~\ref{theorem:structural_properties}, we require the following claims. The next fact follows immediately from the definition of agreement and will be utilized in the subsequent proofs of the claims (cf.~\cite{cohen2021correlation}).

\begin{fact}\label{fact:agreement_facts}
    If $u,v$ are in $i\beta$-agreement, for some $1 \leq i < \frac{1}{\beta}$, then
    \[
    (1-i\beta)d(u) \leq d(v) \leq \frac{d(v)}{1-i\beta}
    \]
\end{fact}

\begin{restatable}{claim}{clustersDisjoint}
\label{claim:clusters_disjoint}
Suppose $(1-3\beta-1.2\epsilon)(1-3\beta) > \frac{1}{2}$.
Assume $u_1, u_2$ are two vertices for which $H(u_1)$ and $H(u_2)$ both return true. 
If $u_2$ is not part of $\agreez{u_1}$, then the sets $\agreez{u_1}$ and $\agreez{u_2}$ are disjoint.
\end{restatable}

\begin{proof}
Let $v$ be a vertex common in both $\agreez{u_1}$ and $\agreez{u_2}$.
Since $u_1$ is $1.2\epsilon$-heavy and in $3\beta$-agreement with $v$, we have that $v$ has at least $(1-3\beta-1.2\epsilon)d(u_1)$ neighbors that are in $\beta$-agreement with $u_1$.
Similarly, $v$ has $(1-3\beta-1.2\epsilon)d(u_2)$ neighbors that are in $\beta$-agreement with $u_2$.
Using Fact~\ref{fact:agreement_facts}, both are non-less than $(1-3\beta-1.2\epsilon)(1-3\beta)d(v)$ and by assumption this is greater than $\frac{1}{2}d(v)$. Consequently, there is a vertex $w$ in $\beta$-agreement with both $u_1$ and $u_2$.

Now, by the triangle inequality we get that $u_2$ is contained in $\agreez{u_1}$:
\begin{align*}
|N(u_1)\triangle N(u_2)|
&< |N(u_1)\triangle N(w)| + |N(w)\triangle N(u_2)| \\
&< \beta\max\{d(u_1), d(w)\} + \beta\max\{d(w), d(u_2)\} \leq 2.4\beta \max\{d(u_1), d(u_2)\}
\end{align*}

\noindent Where the last inequality follows from Fact~\ref{fact:agreement_facts} and that $\beta \leq \frac{1}{6}$.
\end{proof}

\begin{restatable}{claim}{clustersDense}
\label{claim:clusters_dense}
    Suppose $1.2\epsilon \leq 1/3-6\beta$. Every cluster $C$ returned by Algorithm~\ref{algorithm:simplified_structural_clustering} is everywhere dense.
\end{restatable}

\begin{proof}
Consider the cluster $C = \agreez{h}$ created from the $1.2\epsilon$-heavy vertex $h$. We know that every $u\in C$ is in $3\beta$-agreement with $h$. Also by Definition~\ref{definition:heavy} of heavy vertices, $h$ has at most $1.2\epsilon$ fraction of its neighbors outside $C$, hence:
\[
|N(h)\cap N(u)\cap C|
\geq |N(h)\cap N(u)| - |N(h)\setminus C|
\geq (1-3\beta-1.2\epsilon) d(h)
\]
This implies:
\begin{equation}\label{equation:cluster_dense}
d(u, C)
\geq |N(h)\cap N(u)\cap C|
\geq (1-3\beta-1.2\epsilon) d(h)
\end{equation}

Next, we show $d(h)$ is an upper bound for the size of the component $|C|$.
To this end, consider the set of vertices $B = N(h)\cap C$ and the set of edges $E$ between $B$ and $C\setminus B$. Every vertex $u\in C\setminus B$ outside of $B$ is adjacent to at least $|B\cap N(u)| \geq (1-3\beta-1.2\epsilon) d(h)$ vertices inside of $B$ and thus $|E| \geq |C\setminus B|\ (1-3\beta-1.2\epsilon)d(h)$. Moreover, every vertex $u\in B$ inside of $B$ is adjacent to at most $3\beta \max(d(h), d(u)) \leq 3\beta d(h)/(1-3\beta)$ vertices outside of $B$, as deduced from Fact~\ref{fact:agreement_facts}. It follows that $|E| \leq |B|\ 3\beta d(h)/(1-3\beta)$. 
By combining both inequalities we get:

\[
|C\setminus B|
\leq \frac{3\beta}{(1-3\beta)(1-3\beta-1.2\epsilon)} |B|
< \frac{3\beta}{1-6\beta-1.2\epsilon} d(h)
\]
Now, by adding up $|C\setminus B|$ with $|B|$ we obtain an upper bound on $|C|$ in terms of $d(h)$.
\[
|C| = |C\setminus B| + |B|
< \frac{3\beta}{1-6\beta-1.2\epsilon} d(h) + d(h)
= \frac{1-3\beta-1.2\epsilon}{1-6\beta-1.2\epsilon} d(h)
\]
Together with Equation~\ref{equation:cluster_dense}, we achieve the desired result.
\[
d(u, C)
\geq (1-3\beta-1.2\epsilon) d(h)
> (1-6\beta-1.2\epsilon) |C|
\geq \frac{2}{3} |C|
\]
\end{proof}

\begin{restatable}{claim}{importantsAgree}
\label{claim:importants_agree}
    Suppose $0.8\beta \geq 2\epsilon\frac{2-\epsilon}{1-\epsilon}$. Let the pair of vertices $u, v$ belong to the same important group, then $u, v$ are in $0.8\beta$-agreement.
\end{restatable}

\begin{proof}
Suppose $u,v$ belong to the same important group $C$. Through the properties of important groups, we get that, (i) both $u,v$ are adjacent to at least $(1-\epsilon)|C|$ vertices of $C$, and thus disagree on at most $2\epsilon|C|$ vertices inside of $C$. (ii) $u,v$ are adjacent to at most $\epsilon d(u), \epsilon d(v)$ vertices not in $C$, respectively.  In total they disagree on at most:

\begin{multline*}
    |N(u) \triangle N(v)| \leq 2\epsilon|C| + \epsilon(d(u)+d(v))
    \leq \frac{2\epsilon}{1-\epsilon} \max\{d(u), d(v)\} + 2\epsilon \max\{d(u), d(v)\} \\
    \leq 2\epsilon \frac{2-\epsilon}{1-\epsilon} \max\{d(u), d(v)\}
    \leq 0.8\beta \max\{d(u), d(v)\}
\end{multline*}

\noindent Where the second inequality follows from
Definition~\ref{definition:important_group}, a vertex in an important group has a degree that is at least $(1-\epsilon)$ fraction of $C$, that is, for any $u \in C$, $d(u) \geq (1-\epsilon)|C|$.
\end{proof}

\begin{restatable}{claim}{disjointsDisagree}
\label{claim:disjoints_disagree}
    Suppose $2\epsilon \leq 1-3\beta$. Let $u, v$ belong to two disjoint important groups, then $u, v$ are not in $3\beta$-agreement.
\end{restatable}

\begin{proof}
Say that $u, v$ belong to two disjoint important groups $C_u, C_v$, respectively. Then by Definition~\ref{definition:important_group}, $u$ has at least $(1-\epsilon)$ fraction of his neighbors in $C_u$, whereas $v$ has at most $\epsilon$ fraction of his neighbors in $C_u$, which means that $u, v$ disagree on at least $(1-\epsilon)d(u)-\epsilon d(v)$ neighbors inside $C_u$. Similarly, $u, v$ disagree on at least $(1-\epsilon)d(v)-\epsilon d(u)$ neighbors inside $C_v$. Overall, the difference in their neighborhoods is:
\[
|N(u)\triangle N(v)|
\geq (1-2\epsilon)(d(u)+d(v)) > (1-2\epsilon) \max \{ d(u), d(v) \}
\]
The claim now follows directly from the Definition~\ref{definition:agreement} together with the assumption that $1-2\epsilon \geq 3\beta$.
\end{proof}

We are finally ready to prove Lemma~\ref{theorem:structural_properties}.

\begin{proof}

Following Claim~\ref{claim:clusters_disjoint} and Claim~\ref{claim:clusters_dense} the clusters returned by the algorithm are disjoint and everywhere dense. Next, we show the properties related to important groups. First, let $u$ be a vertex in some important group $C$, then $u$ has at least a $1-\epsilon$ fraction of its neighbors within $C$, and according to Claim~\ref{claim:importants_agree}, it is in $0.8\beta$-agreement with all vertices in $C$. Thus, every vertex that belongs to an important group is $\epsilon$-heavy, which also implies that any such vertex is part of a non-singleton cluster. 

Now, assume $u, v$ belong to the same important group $C$, and that $u$ belongs to a cluster $\agreez{h}$ created in step~\ref{alg:cluster_creation} of the algorithm, we will show that $C \subseteq \agreez{h}$.
Because $u, v$ are in $0.8\beta$-agreement, $u$ also belongs to the set $\agreez{v}$. However, the intersection of $\agreez{h}$ and $\agreez{v}$ at vertex $u$ implies, based on Claim~\ref{claim:clusters_disjoint}, that vertex $v$ necessarily belongs to the cluster $\agreez{h}$.

It remains to prove that if $u$ and $v$ belong to disjoint important groups, they cannot be part of the same cluster. Suppose, contrarily, that both $u$ and $v$ belongs to $\agreez{h}$. Note that $u$ is $\epsilon$-heavy as it belongs to some important group. As such, since $u$ in $\agreez{h}$, the sets $\agreez{u}$ and $\agreez{h}$ are not disjoint (as both contains $u$). According to Claim~\ref{claim:clusters_disjoint}, as both $u,h$ are $\epsilon$-heavy, $h$ must also be in $\agreez{u}$. Similarly, $h$ belongs to $\agreez{v}$ as well. Thus, $\agreez{u}$ and $\agreez{v}$ intersect at $h$. However, by Claim~\ref{claim:clusters_disjoint}, this intersection implies that $u$ and $v$ must be in $3\beta$-agreement, contradicting Claim~\ref{claim:disjoints_disagree}.
\end{proof}

\subsubsection{S-Structural Clustering}
\label{sec:Sstructure}

So far we have provided an algorithm for $V$-Structural Clustering. However, what we really need for $\ell_0$ Best-Fit Ultrametrics is the more general problem of $S$-Structural Clustering as defined in Lemma~\ref{lemma:sclustering}. Note that it is different from Agreement Correlation Clustering on the induced subgraph $G[S]$, because $S$-Structural Clustering also considers edges with only one endpoint in $S$.

\begin{restatable}{lemma}{restSstructural}
\label{lemma:sclustering}
    Suppose $\beta = 5\epsilon(1+\epsilon)$ for a small enough parameter $\epsilon \leq 1/95$.
    For every $S \subseteq V$
    we can output a set of clusters $\clusters$. This clustering ensures that for every important group of vertices $C' \subseteq S$, there is a cluster $C \in \clusters$ such that $C' \subseteq C$, and $C$ does not intersect any other important groups of vertices contained in $S \setminus C'$. Moreover, every
    cluster $C \in \clusters$ is everywhere dense.
\end{restatable}

In the rest of this section we provide a construction of $S$-Structural Clustering by reducing the problem to $V$-Structural Clustering. We start by generalizing the definitions \ref{definition:agreement} and \ref{definition:heavy} presented earlier:

\begin{definition}[subset agreement]
\label{definition:subset_agreement}
We say that two vertices $u, v\in S$ are in $\beta$-agreement inside $S$ if $\abs{N(u)\triangle N(v)} + 2\abs{N(u)\cap N(v)\cap \overline{S}} < \beta \max\{d(u), d(v)\}$, which means that $u, v$ share most of their neighbors inside $S$. Denote by $A_S(u)$ the set of vertices that are in $\beta$-agreement with $u$ inside $S$.
\end{definition}

\begin{definition}[subset heavy]
\label{definition:subset_heavy}
We say that a vertex $u\in S$ is $\epsilon$-heavy inside $S$ if $\abs{N(u) \setminus A_S(u)} < \epsilon d(u)$, which means that most of its neighbors are in agreement with $u$ inside $S$. Denote by $H_S(u)$ the $\epsilon$-heaviness indicator of vertex $u$ inside $S$.
\end{definition}

Note that definitions \ref{definition:subset_agreement} and \ref{definition:subset_heavy} are more general than definitions \ref{definition:agreement} and \ref{definition:heavy}, since we can derive the latter ones by substituting $S = V$. The extra term $2\abs{N(u)\cap N(v)\cap \overline{S}}$ has an intuitive meaning that will become clear later during the reduction. Similarly to the previous section we need to approximate the new sets $A_S(u), A_S^3(u)$ and $H_S(u)$, which we do in \Cref{section:semi_agreement_cc}.

We finally prove Lemma \ref{lemma:sclustering} by reducing $S$-Structural Clustering, for a set $S$ in the correlation clustering instance $G$, to $V$-Structural Clustering, in a specially constructed instance $G_S$. The instance $G_S$ can be seen as a transformation of $G$ that preserves all internal edges within $S$ and replaces all external neighbors $v \in \overline{S}$ of internal vertices $u \in S$ with dummy vertices, ensuring that our subset-specific definitions of agreement \ref{definition:subset_agreement} and heaviness \ref{definition:subset_heavy} applied to $G$ correspond exactly to the original definitions \ref{definition:agreement} and \ref{definition:heavy} applied to $G_S$. Therefore, to produce the desired $S$-Structural Clustering, it suffices to run Algorithm \ref{algorithm:simplified_structural_clustering} on $S$ using $A_S(u), A_S^3(u), H_S(u)$ as the parameters. The full proof is presented below.

\begin{proof}
    Denote by $\clusters_S$ the clustering of $S$ returned by Algorithm \ref{algorithm:simplified_structural_clustering} when executed over the vertex set $S$ using $A_S(u), A_S^3(u), H(u)$ as parameters. These parameters refer to the $S$-subset agreement sets \ref{definition:subset_agreement} and $S$-subset heaviness indicator \ref{definition:subset_heavy} calculated for each vertex $u\in S$ given our correlation clustering instance $G$.
    Also denote by $G_S$ a new correlation clustering instance that contains all the vertices in $S$ and some additional dummy vertices. Given the subset $S$, the instance $G_S$ can be seen as a transformation of $G$ with the following steps. First insert all the edges of $G$ with internal endpoints $u, v\in S$ to $G_S$. Second, consider all the edges of $G$ with an internal endpoint $u\in S$ and an external endpoint $v\in \overline{S}$. If $u$ has any neighbor other than $v$ in the instance $G$ (that is $d(u) > 2$), then create a dummy vertex $u_v$ and insert the edge with endpoints $u, u_v$ to $G_S$.

    Next we need to notice that $A_S(u)$, which is a $\beta$-agreement set for $G$ inside $S$ according to Definition \ref{definition:subset_agreement}, is also a $\beta$-agreement set for $G_S$ according to Definition \ref{definition:agreement}. By definition $A_S(u)$ is calculated for every internal vertex $u\in S$ of $G_S$. For the sake of clarity we also define $A_S(u) = \{u\}$ for every dummy vertex $u\in \overline{S}$ of $G_S$.
    Indeed consider any dummy vertex $u_v\in \overline{S}$ of $G_S$, which is only adjacent to its internal vertex $u\in S$, by construction of $G_S$. Since $u$ has at least some neighbor other than $u_v$, then $u_v$ is not in $\beta$-agreement with $u$ by Definition~\ref{definition:agreement}. But $u_v$ is not in $\beta$-agreement with any other vertex $w\neq u$ since $u$ can be their only common neighbor ($N_S(u_v) = \{u_v, u\}$). It now suffices to show that any pair of (non singleton) internal vertices $u, v\in S$ of $G_S$ is in $\beta$-agreement according to Definition~\ref{definition:agreement} if and only if $u, v\in S$ of $G$ are in $\beta$-agreement inside $S$ according to Definition~\ref{definition:subset_agreement}. But this is true since, by construction of $G_S$, the degrees $d(u), d(v)$ in $G$ are equal to the degrees $d_S(u), d_S(v)$ in $G_S$ and that $|N_S(u)\triangle N_S(v)| = |N(u)\triangle N(v)| + 2|N(u)\cap N(v)\cap \overline{S}|$.
    
    Now we will prove that $\clusters_S$ along with the singleton vertices $u_v\in \overline{S}$ of instance $G_S$ constitute a $V$-Structural Clustering for $G_S$, that is a clustering satisfying the structural properties of Lemma~\ref{theorem:structural_properties}. Using the same argument as with $A_S(u)$, we see that $A_S^3(u)$ is a $3\beta$-agreement set of $G_S(u)$ and using the definitions \ref{definition:heavy} and \ref{definition:subset_heavy} we see that $H_S(u)$ is an $\epsilon$-heaviness indicator of $G_S$. Just for the sake of clarity we also define $A_S^3(u) = \{u\}$ and $H_S(u) = \text{false}$ for every dummy vertex $u\in \overline{S}$ of $G_S$.
    Given that the parameters $A_S(u), A_S^3(u), H_S(u)$ are some proper agreement sets \ref{definition:agreement} and heaviness indicators \ref{definition:heavy} for the instance $G_S$, then running Algorithm \ref{algorithm:simplified_structural_clustering} in the entire vertex set of $G_S$ produces a $V$-Structural Clustering for $G_S$. But during the execution of the Algorithm \ref{algorithm:simplified_structural_clustering}, every dummy vertex $u_v\in \overline{S}$ of $G_S$ will be eventually ignored. Indeed the algorithm will never iterate through $u_v$ as $H(u_v) = \text{false}$ and the algorithm will never include $u_v$ in a non trivial cluster as there is no set $u_v\in A_S^3(w)$. So it would be equivalent to first iterate through $S$ to create clusters $\clusters_S$ and later iterate through $\overline{S}$ to create the remaining singleton clusters.
    
    Finally we are ready to prove that $\clusters_S$ is the required $S$-Structural Clustering of $G$, that is a clustering of $S$ satisfying the structural properties of Lemma \ref{lemma:sclustering}. We start by observing that that a group of vertices $C'\subseteq S$ is important in $G$ if and only if it is important in $G_S$. By construction of $G_S$ any (non singleton) internal vertex $u\in S$ has the same total degree $d(u) = d_S(u)$ and set of internal neighbors $N(u) = N_S(u)$ in both graphs $G, G_S$. And since $C'\subseteq S$, every vertex $u\in C'$ is adjacent to the same fraction of the vertices in $|C'|$ and the same fraction of its neighbors outside of $C'$ in both graphs $G, G_S$, which proves the equivalence of Definition~\ref{definition:important_group} in the two graphs.
    Now we know that every important group of vertices $C'\subseteq S$ of $G$ is also important in $G_S$ and subsequently from Lemma~\ref{theorem:structural_properties} there is a cluster $C\in \clusters_S$ such that $C'\subseteq C$. Also, cluster $C$ does not intersect any other important groups of vertices of $G$ contained in $S\setminus C'$, since otherwise it would intersect with the respective disjoint important groups of vertices of $G_S$, which would contradict Lemma~\ref{theorem:structural_properties}. Lastly, any cluster $C\in \clusters_S$ is everywhere dense in both $G$ and $G_S$, since $C\subseteq S$ has the exact same internal edges by construction of $G_S$.
\end{proof}

\newcommand{\vagagr}{\abs{N(u)\cap N(v)} +
\abs{N(u)\triangle N(v)}- \abs{N(u)\cap N(v)\cap S}}

\subsubsection{Computing Agreements}\label{section:semi_agreement_cc}

Building on the algorithm from \Cref{sec:Sstructure}, we now describe its adaptation to the semi-streaming model. Since this algorithm only requires computing approximations to $\beta$-agreements in $S$ and heaviness queries, we need to prove the following lemma:
\begin{restatable}{lemma}{subsetstructuralimpl}\label{lemma:subset_structural_impl}
The following statements hold with high probability:

\begin{enumerate}
    \item For a given $\gamma \in \{\beta, 3\beta\}$ and every $u,v \in S$, we can output `yes' if $\abs{N(u)\triangle N(v)} + 2\abs{N(u)\cap N(v)\cap \overline{S}} < 0.8\gamma \max\{d(u), d(v)\}$ and `no' if $\abs{N(u)\triangle N(v)} + \abs{N(u)\cap N(v)\cap \overline{S}} \geq \gamma \max\{d(u), d(v)\}$.
    \item For every $u \in S$, we can output `yes' if $\abs{N(u) \setminus A_S(u)} < \epsilon d(u)$ and `no' if $\abs{N(u) \setminus A_S(u)} > 1.2 \epsilon d(u)$.
\end{enumerate}
\end{restatable}

Before proving the lemma we state the following corollary which is a direct consequence of Claim~\ref{clm:sketch1nb2}.

\begin{restatable}{corollary}{closesketch}\label{coro:closesketch}
With high probability, for each vertex \(v\) and each weight $w$, there exists $\tilde{w}^v_s$, such that $\abs{N_w(v) \triangle N_{\tilde{w}^v_s}(v)} \leq 5\zeta \abs{N_w(v)}$.
\end{restatable}

\begin{proof}
If $N_w(v) \subseteq  N_{\tilde{w}^v_s}(v)$, then by Claim~\ref{clm:sketch1nb2},
$\abs{N_w(v) \triangle N_{\tilde{w}^v_s}(v)} \leq \abs{N_{\tilde{w}^v_s} \setminus  N_w(v)} \leq 5\zeta \abs{ N_w(v)}$. Else, $\abs{N_w(v) \triangle N_{\tilde{w}^v_s}(v)} \leq \abs{N_w(v) \setminus N_{\tilde{w}^v_s}} \leq \zeta \abs{N_w(v)}$. 
\end{proof}

\begin{proof}
We are now ready to prove \Cref{lemma:subset_structural_impl}. For the first item, let $d(v) \geq d(u)$ and consider the 3 possible cases: (i) $d(u),d(v) \leq 2\log^4 n$, (ii) $d(u) \leq \log^4 n$ and $d(v) > 2\log^4 n$, and, (iii) \(d(u),d(v) \geq \log^4 n\).

In the first case case the entire neighborhoods of $u$ and $v$ are known and the query can be computed precisely. Whereas in the second case, $\abs{N(u) \triangle N(v)} \geq d(v)-d(u) \geq \frac{d(v)}{2}$, this implies that, $u,v$ cannot satisfy the conditions required for a `yes' instance, and we report `no'. Consequently, we remain with the third case which will require the sketching scheme outlined in Section~\ref{section:sketches}.

Let $s_u$ and $s_v$ be the sizes reported by Claim~\ref{clm:sketch1nb2} for $u$, and $v$ respectively. Then, we have that w.h.p, $\frac{s_v}{1+5\zeta} \leq d(v) \leq \frac{s_v}{1-\zeta}$, and similarly for $d(u)$. Thus:

\[   \abs{N(u)\triangle N(v)} \geq d(v)-d(u) = d(v)\big(1-\frac{d(u)}{d(v)}\big) \geq d(v)\big(1-\frac{(1+5\zeta)s_u}{(1-\zeta)s_v}\big)
\]
Consequently, if $1-\frac{(1+5\zeta)s_u}{(1-\zeta)s_v} > 0.8\gamma$, then $\abs{N(u) \triangle N(v)} > 0.8\gamma d(v)$, and $u,v$ cannot satisfy the conditions required for a `yes' instance in this lemma, and thus we report `no'.

Else, it is the case that $\frac{s_u}{s_v} \geq (1-0.8\gamma)\frac{1-\zeta}{1+5\zeta} > \frac{1}{2}$, where the last inequality follows by selecting sufficiently small values for $\gamma$ and $\zeta$.

Note that, $\abs{N(u)\triangle N(v)} = d(u)+d(v)-2\abs{N(u)\cap N(v)}$, hence:
\begin{equation}\label{eq:iden}
\begin{aligned}
\abs{N(u)\triangle N(v)} + 2\abs{N(u)\cap N(v)\cap \overline{S}} 
&= d(u)+d(v) - 2\abs{N(u) \cap N(v) \cap S}    
\end{aligned}
\end{equation}
Using $s_u,s_v$ we can approximate $d(u),d(v)$, respectively, and obtain:

\begin{equation}\label{eq:est}
\begin{aligned}
    &\frac{s_u}{1+5\zeta} + \frac{s_v}{1+5\zeta} - 2\abs{N(u) \cap N(v) \cap S} \\
    &\qquad\leq \abs{N(u)\triangle N(v)} + 2\abs{N(u)\cap N(v)\cap \overline{S}} \\
    &\qquad\leq \frac{s_u}{1-\zeta} + \frac{s_v}{1-\zeta} - 2\abs{N(u) \cap N(v) \cap S}
\end{aligned}
\end{equation}

Thus, to estimate $\abs{N(u)\triangle N(v)}+ 2\abs{N(u)\cap N(v)\cap \overline{S}} $, it is enough to estimate $\abs{N(u) \cap N(v) \cap S}$. W.l.o.g. assume $s_u\ge s_v$. We will consider the sketches ${\Sk}^{v}_{s_v,s_v}$ and ${\Sk}^{u}_{s_u,s_v}$ (otherwise, consider the sketches ${\Sk}^{v}_{s_v,s_u}$ and ${\Sk}^{u}_{s_u,s_u}$). 
Using Corollary~\ref{coro:closesketch}, $N(u)$ and $N_{\tilde{w}^u_{s_u}}(u)$ disagree on at most  $5\zeta d(u)$ elements where $5\zeta d(u) \leq 5\zeta d(v) \leq 5\zeta \frac{s_v}{1-\zeta}$, and similarly for $N(v)$ and $N_{\tilde{w}^v_{s_v}}(v)$. Let $M=N_{\tilde{w}^v_{s_v}}(v) \cap N_{\tilde{w}^u_{s_u}}(u)$, then:

\begin{equation}\label{eq:s_estimate}
 \abs{N(u) \cap N(v) \cap S} -10\zeta \frac{s_v}{1-\zeta} \leq
\abs{M \cap S}
\leq \abs{N(u) \cap N(v) \cap S}  +10\zeta \frac{s_v}{1-\zeta}
\end{equation}

Define a random variable $X^S_{u,v}=\abs{{\mathcal{N}}^{v}_{s_v} \cap {\mathcal{N}}^{u}_{s_v} \cap S}$. 
Recall that these sketches are constructed using the random set $R_{s_v} \subseteq V$, where each vertex of $V$ is sampled independently at random with probability $\frac{\log^2 n}{s_v}$.
By linearity of expectation we have:
\begin{align*}
\mathbb{E}[X^S_{u,v}] = \frac{\log^2 n}{s_v}\abs{M \cap S}  
\end{align*}
We apply Chernoff bound to obtain w.h.p 
$(1-\zeta)\frac{s_v}{\log^2 n}X^S_{u,v} < \abs{M \cap S} < (1+\zeta)\frac{s_v}{\log^2 n}X^S_{u,v}$.

Combining both Equation~\ref{eq:est} and Equation~\ref{eq:s_estimate} with the bounds on $\abs{M \cap S}$, we get:

\begin{equation}\label{eq:final}
\begin{aligned}
    &2\frac{s_v}{1+5\zeta}-\frac{20\zeta}{1-\zeta}s_v-2(1+\zeta)\frac{s_v}{\log^2 n}X^S_{u,v} \\
    &\qquad\leq \abs{N(u)\triangle N(v)} + 2\abs{N(u)\cap N(v)\cap \overline{S}} \\
    &\qquad\leq (2+5\zeta)\frac{s_v}{1-\zeta}+\frac{20\zeta}{1-\zeta}s_v-2(1-\zeta)\frac{s_v}{\log^2 n}X^S_{u,v}
\end{aligned}
\end{equation}

Observe that by Equation~\ref{eq:final}, for some constant $k$, we can also write:

\[\abs{(\abs{N(u)\triangle N(v)} + 2\abs{N(u)\cap N(v)\cap \overline{S}})-(2s_v-2\dfrac{s_v}{\log^2 n}X^S_{u,v})}\le k\zeta s_v\] 

The lemma now follows by selecting small enough parameter $\zeta$ relative to $\gamma$, namely, $\zeta < \frac{1}{10k} \gamma$. Based on this, we report 'yes' if $2s_v-2\frac{s_v}{\log^2 n}X^S_{u,v} \leq 0.9\gamma s_v$, and 'no' otherwise.

For the second item, if $d(u) \leq 2\log^4 n$ the entire neighborhood is known and the query can be computed precisely. Else, let ${\mathcal{N}}^{u}_{s_u}$ be the vertices defining the edges incident on $u$ in sketch ${\Sk}^{u}_{s_u,s_u}$ and define the random variable $Y_u=|{\mathcal{N}}^{u}_{s_u} \cap S\cap A_S(u)|$. 

Observe that, $\abs{N(u) \setminus A_S(u)} = \abs{N(u)} -\abs{N(u) \cap S \cap A_S(u)}$, and we can write: 
\begin{equation}
    \frac{s_u}{1+5\zeta}-\abs{N(u) \cap S \cap A_S(u)} \leq \abs{N(u) \setminus A_S(u)} \leq \frac{s_u}{1-\zeta}-\abs{N(u) \cap S \cap A_S(u)}
\end{equation}

Similarly to the first part of the lemma we estimate $\abs{N(u) \cap S \cap A_S(u)}$. Using Corollary~\ref{coro:closesketch} we have:
\begin{align}
    \abs{N_{\tilde{w}^u_{s_u}}(u) \cap S \cap A_S(u)}-5\zeta\frac{s_u}{1-\zeta} \leq \abs{N(u) \cap S \cap A_S(u)} \leq \abs{N_{\tilde{w}^u_{s_u}}(u) \cap S \cap A_S(u)}+5\zeta\frac{s_u}{1-\zeta}
\end{align} 

Note that, ${\mathcal{N}}^{u}_{s_u}$ contains a random sample of $N_{\tilde{w}_{s_u}^u(u)}$, where each vertex is sampled with probability $\dfrac{\log^2 n}{s_u}$. Thus, by linearity of expectation, the expected value of $Y_u$ is $\frac{\abs{N_{\tilde{w}_{s_u}^u(u)} \cap S \cap A_S(u)}}{s_u}\log^2 n$. Using Chernoff bound, we obtain w.h.p that:
\begin{equation}
    (1-\zeta)\frac{s_u}{\log^2 n}Y_u < \abs{N_{\tilde{w}_{s_u}^u} \cap S \cap A_S(u)} < (1+\zeta)\frac{s_u}{\log^2 n}Y_u
\end{equation} 
We then report `yes' if $s_u-\frac{s_u}{\log^2 n}Y_u \leq 1.1 \epsilon s_u$, and 'no' otherwise. We conclude that:

\[\abs{\abs{N(u)\setminus A_S(u)} -(s_u-\dfrac{s_u}{\log^2 n}Y_u)}\le k\zeta s_u\] 

For some constant $k$. The lemma now follows by selecting small enough parameter $\zeta$ relative to $\epsilon$, namely, $\zeta < \frac{1}{10k} \epsilon$.
\end{proof}

We can now establish S-Structural Clustering in the semi-streaming model:
\begin{theorem} \label{lem:horizontal}
    Suppose $\beta = 5\epsilon(1+\epsilon)$ for a small enough parameter $\epsilon \leq 1/95$.
    Given access to the sketches of all vertices in $S$, and w.h.p.,
    for every $S \subseteq V$
    we can output a set of clusters $\clusters$. This clustering ensures that for every important group of vertices $C' \subseteq S$, there is a cluster $C \in \clusters$ such that $C' \subseteq C$, and $C$ does not intersect any other important groups of vertices contained in $S \setminus C'$. Moreover, every
    cluster $C \in \clusters$ is everywhere dense. 
\end{theorem}
\begin{proof}
The algorithm outlined in \Cref{sec:Sstructure} only requires the computation of polynomially many approximations to $\beta$-agreements in $S$ and queries of heaviness in $S$; these can be computed with high probability using \Cref{lemma:subset_structural_impl}.
Note that this only requires access to the sketches of vertices in $S$.
The theorem now follows from \Cref{lemma:sclustering}.
\end{proof}

\subsection{Main Algorithm}
\label{section:algorithm}

To run our main algorithm it suffices to obtain access to certain black-boxes established in the previous sections.
Ideally, we would like to have access to a summary of the input distances, to estimations of neighborhood sizes, and to be able to repeatedly compute $S$-Structural Clustering for instances given by an adaptive adversary.
We show that even though we cannot generally guarantee the last requirement, it suffices to guarantee it for a particular (technical) type of adaptive adversary (see \Cref{lem:blackBoxes}).

We first need the following definitions.
For a weight $w \in D$, let $\widehat{w}$ be the smallest weight in $\widetilde{D}$ such that $\widehat{w}>w$.
Similarly, for any $w$ we let $\widecheck{w}$ be the largest value in $\widetilde{D}$ smaller than $w$.
We say that $\widetilde{D}$ is a \emph{compressed set} if for $w\not \in \widetilde{D}$, $\widetilde{D}$ has the property that $d_w(u) \le (1+\delta) d_{\widecheck{w}}(u)$ for all $u$.
Finally let $\widetilde{d_w(u)}$ be a function with $\widetilde{d_w(u)}\in [(1-\lambda)d_w(u), (1+\lambda)d_w(u)]$ for a sufficiently small constant $\lambda$.

Our algorithm (see \Cref{algorithm:old_structural_clustering} for the pseudocode) is a divisive algorithm running $S$-Structural Clustering at each level to divide a cluster.
However, it then performs a different division strategy for the largest cluster.
This different strategy for the largest cluster allows us to guarantee that each vertex only participates in a logarithmic number of $S$-Structural Clustering computations, and is only possible if the size of the largest cluster has not dropped by a constant factor.

More formally, our algorithm takes as argument a set $S$ (initially the whole $V$) and a distance $w$ (initially the maximum distance).
First, it creates a tree-node $A$ at distance $w/2$ from the leaves, whose leaves-descendants are all the vertices in $S$.
Then it uses an $S$-Structural Clustering subroutine, and for each cluster $C'$ with size at most $0.99|S|$ it recurses on $(C',\widecheck{w})$.
The roots of the trees created from each of these recursions then become children of $A$.

Subsequently, for the largest cluster $C$ we perform the following postprocessing:
Let $w' \gets \widecheck{w}$ and $w''\gets \widecheck{w'}$.
\begin{itemize}
    \item If there are at most $0.99|S|$ vertices $u$ in $S$ with large estimated degree $\widetilde{d_{w''}(u)}$ (larger than $0.66 |S|$), then we recurse on $(C,w')$; the root of the tree created from this recursion becomes a child of $A$.
    \item Otherwise, we let $R$ contain the vertices $u$ whose estimated degree $\widetilde{d_{w''}(u)}$ is small (less than $0.65 |S|$), and recurse on $(R,w')$.
    The root of the tree created from this recursion (let us call it $A'$) becomes a child of $A$, and then we update $A \gets A'$.
    Finally, we repeat the postprocessing again (but this time on $C\setminus R$ (instead of $R$) at level $\widecheck{w}$ (instead of $w$)).
\end{itemize}

\begin{algorithm}[]
\caption{\textsc{$\ell_0(S, w)$}}
\label{algorithm:old_structural_clustering}
\begin{algorithmic}[1] \label{alg:HierL0}
    \State Mark $S$ at level $w$ as a core cluster
    \If{$|S|\le 1$} \Return \EndIf

    \State obtain $\calC = \set{C_1, \ldots, C_k}$ using an $S$-Structural Clustering subroutine on $(V,E_{\widecheck{w}})$
    \ForAll{$u,v$ in different clusters of $\calC$}
        $T(uv) \gets w$ \label{line:L0dist1}
    \EndFor
    \ForAll{$C' \in \calC$ with $|C'| \le 0.99|S|$}
        \textsc{$\ell_0(C', \widecheck{w})$} \label{line:L0rec1}
    \EndFor
    \If{$\exists C\in \calC$ with $|C| > 0.99|S|$}
        \State $w' \gets \widecheck{w}$, $w'' \gets \widecheck{w'}$
        \While{$|C| > 0.99|S|$ and $|\set{u\in C \mid \widetilde{d_{w''}(u)} > 0.66 |S|}| > 0.99|S|$} \label{line:whileLoopL0}
            \State $R \gets \set{u\in C \mid \widetilde{d_{w''}(u)} < 0.65 |S|}$ \label{line:L0R}
            \ForAll{$u\in R, v\in C\setminus R$}
            \State $T(uv) \gets w'$ \label{line:L0dist2}
            \EndFor
            \State \textsc{$\ell_0(R, w')$} \label{line:L0rec2}
            \State $C \gets C\setminus R, w' \gets \widecheck{w'}, w'' \gets \widecheck{w''}$
        \EndWhile
        \State \textsc{$\ell_0(C, w')$} \label{line:L0rec3}
    \EndIf
\end{algorithmic}
\end{algorithm}

The idea of the postprocessing is that nodes whose degree drops significantly cannot be in a huge cluster without a big cost.
The challenging part is showing that keeping the rest of the nodes in $C$ is sufficient.

In the rest of this section we provide the proof of our main result (\Cref{theorem:lzerofit}) along with its required lemmas.
We remind the reader that even though (for simplicity) \Cref{algorithm:old_structural_clustering} explicitly stores the output distance between every pair of vertices, we cannot afford to do that in the semi-streaming model.
That is why, in the proof of Theorem~\ref{theorem:lzerofit}, we show how we can implicitly represent all these distances by storing a tree.
From this point on, we let $T = \ell_0(V, w_{max})$ be the output of \Cref{alg:HierL0}.

We first provide two results: $T$ is a valid ultrametric, and the depth of the recursion of \Cref{alg:HierL0} is $O(\log{n})$.
Informally, the latter is crucial in order to limit the dependencies across different recursive calls, which in turn allows us to treat different recursive calls as independent from each other.
Of course the components of the algorithm guaranteeing the $O(\log{n})$ recursion depth also make the analysis of the algorithm different.

\begin{restatable}{lemma}{isUltrametric}\label{lem:isUltrametric}
$T = \ell_0(V, w_{max})$ is a valid ultrametric.
\end{restatable}
\begin{proof}
We inductively prove that for any three vertices $u_1,u_2,u_3$, the strong triangle inequality (characterizing ultrametrics) $T(u_1u_2) \le \max\set{T(u_1u_3), T(u_2u_3)}$ is satisfied.
It trivially follows if $|S|=1$.

Otherwise, for any three vertices $u_1,u_2,u_3$, if not all $3$ of them are in the same cluster of $\calC$, then by Line~\ref{line:L0dist1} at least two pairs have distance $w$ in $T$.
The other pair cannot get distance larger than $w$, thus the strong triangle inequality is satisfied.

If all $3$ of them are in a cluster of $\calC$ with size at most $0.99|S|$, then our claim holds inductively, when we recurse in Line~\ref{line:L0rec1}.

If all $3$ of them are in the unique cluster of $\calC$ with size greater than $0.99|S|$, then either all $3$ of them stay in $C$ by the end of the while-loop (and thus inductively our claim holds when recursing in Line~\ref{line:L0rec3}), or there is a first time when one of them (say $u_1$) is in $R$.
Now:
\begin{itemize}
    \item If at the same time all of them are in $R$, inductively our claim holds when we recurse in Line~\ref{line:L0rec2}.
    \item Else, if one more (say $u_2$) is in $R$, then $T(u_1u_3) = T(u_2u_3) = w'$, and $T(u_1u_2)$ can be at most $w'$, therefore the strong triangle inequality is satisfied.
    \item Otherwise $T(u_1u_2) = T(u_1u_3) = w'$, and $T(u_2u_3)$ can be at most $w'$, therefore the strong triangle inequality is satisfied.
\end{itemize}
\end{proof}

To analyze the approximation factor of our algorithm, we first define a tree $OPT'$ that is an $O(1)$ approximation of an optimal tree $OPT$, but has more structure.
We then show that $T$ is an $O(1)$ approximation of $OPT'$, and therefore an $O(1)$ approximation of $OPT$ as well.

\begin{restatable}{lemma}{recDepth}\label{lem:recDepth}
In \Cref{alg:HierL0}, for any given $u\in V$ we have that the number of recursive calls \textsc{$\ell_0(S, w)$} with $u\in S$ are $O(\log{n})$.
\end{restatable}
\begin{proof}
If we recurse in Line~\ref{line:L0rec1}, or in Line~\ref{line:L0rec3} after having $|C| \le 0.99|S|$, the size of the vertex-set $S$ drops by a constant factor.

If we recurse in Line~\ref{line:L0rec3} while $|C| > 0.99|S|$, then it holds that $|\set{u\in C \mid \widetilde{d_{w''}(u)} > 0.66 |S|}| \le 0.99|S|$.
When in the next recursion call we run $S$-Structural Clustering, we have that for any cluster $C$ and any vertex $u$ it holds $d_{\widecheck{w'}}(u) \ge \frac{2}{3} |C|$, or equivalently $|C| \le 1.5 d_{\widecheck{w'}}(u)$.
Now if $C$ only contains vertices $v$ with $\widetilde{d_{w''}(v)} > 0.66 |S|$, we get $|C| \le 0.99|S|$.
Otherwise it contains a vertex $u$ with $\widetilde{d_{w''}(u)} \le 0.66 |S|$, which implies $d_{\widecheck{w'}}(u) \le 0.66|S| / (1-\lambda)$, and thus $|C| \le 0.99|S| / (1-\lambda)$, which is less than $0.9999|S|$ for sufficiently small $\lambda$.
In all cases, the size of $C$ drops by a constant factor, and therefore all subsequent recursive calls are called with a vertex-set which is a constant factor smaller than $S$.

If we recurse in Line~\ref{line:L0rec2}, it holds that the size of $R$ is at most $0.01|S|$, as $R$ only contains vertices $u$ with $\widetilde{d_{w''}(u)} < 0.65 |S| < 0.66 |S|$.

In all cases, after at most $2$ recursive calls, the size of the vertex-set argument drops by a constant factor, and thus the claim follows.
\end{proof}

\paragraph*{Obtaining $OPT'$} Let us now describe how to obtain $OPT'$, given $OPT$.
To make the exposition easier, we define some intermediary trees that are also constant factor approximations to $OPT$.

We first use the transformation from \cite{cohen2022fitting} on $OPT$, to acquire $OPT'_1$.
We write $OPT'_1 = f(OPT)$ to denote this transformation.
It works as follows:
We first set $OPT'_1 = OPT$, and proceed top-down.
If a cluster $C$ has way too many missing internal edges ($|\set{uv | u,v \in C, D(uv) \ge w}| > \frac{\epsilon^2 |C|^2}{12.5}$), or way too many outgoing edges ($|\set{uv | u \in C, v \not \in C, D(uv) < w}| > \frac{\epsilon^2 |C|^2}{12.5}$) we set $OPT'(uv)=w$ for all $u,v\in C$.
Viewing $OPT'_1$ as a tree, this corresponds to replacing the subtree rooted at $C$ with a star, effectively separating all vertices in $C$ into singletons in all lower levels.
Then we again proceed top-down; as long as there exists a vertex $u$ in a non-singleton cluster $C$ at level $w$, with more than an $\epsilon$ fraction of its neighbors outside $C$, or less than $(1-\epsilon)$ of its neighbors inside $C$, we set the distance of $u$ to $w$.
Viewing $OPT'_1$ as a tree, this corresponds to removing $u$ from all lower level clusters, effectively making it a singleton in all lower levels.

From the construction of $OPT'$ we get the following properties:

\begin{lemma} \label{lem:transformationVincent}
Given an ultrametric $U$, let $U' = f(U)$.
It holds that:
\begin{itemize}
    \item $\|U'-D\|_0 = O(\|U-D\|_0)$.
    \item every cluster in $U'$ is either an important cluster in its respective level or a singleton.
    \item if $U'(u'v') = w$ for any $u',v'$, then there exist $u,v$ such that $U(uv) = w$.
\end{itemize} 
\end{lemma}
\begin{proof}
The first two claims follow from Lemma $3.3$ of \cite{cohen2022fitting}, and the third claim follows directly by the construction of $U'$.
\end{proof}

We then create $OPT'_2$, which only has distances that are in $\widetilde{D}$, by modifying $OPT'_1$.
If $OPT'_1(uv)\not\in \widetilde{D}$, we set $OPT'_2(uv) = \widehat{OPT'_1(uv)}$.
Otherwise, we set $OPT'_2(uv) = OPT'_1(uv)$.
We say we \emph{destroy} a cluster in an ultrametric tree by connecting its children to its parent and then removing said cluster (note that destroying a cluster in an ultrametric tree preserves the ultrametric property).
Viewing $OPT'_2$ as a tree, our transformation corresponds to destroying all clusters at levels that are not in $\widetilde{D}$, one by one.
Finally, we apply the transformation of \Cref{lem:transformationVincent} again, to get $OPT' = f(OPT'_2)$.
Given the tree view of $OPT'$, it is straightforward to verify that it is indeed an ultrametric, as $OPT$ is also an ultrametric.

We now establish structural properties of $OPT'$. Recall that a cluster $C\subseteq V$ is important by Definition~\ref{definition:important_group} if each vertex $u\in C$ is adjacent to at least $(1-\epsilon)$ fraction of vertices in $C$ while having at most $\epsilon$ fraction of its neighbors outside $C$. This definition leads naturally to the following lemma:

\begin{lemma} \label{lem:importantDegrees}
Let $C$ be an important cluster, and $u$ be a vertex in $C$.
Then $(1-\epsilon)|C| \le d(u) \le |C|/(1-\epsilon)$.
Equivalently, for any $u\in C$ we have $(1-\epsilon)d(u) \le |C| \le d(u) / (1-\epsilon)$.
\end{lemma}
\begin{proof}
Directly by the definition of an important cluster~\ref{definition:important_group}, $u$ is connected with $(1-\epsilon)|C|$ vertices in $C$, therefore $(1-\epsilon)|C| \le d(u)$.
On the other hand, $u$ can be connected with all vertices in $C$, and only an $\epsilon$ fraction of its edges can be out of $C$.
Therefore $d(u) \le |C| + \epsilon d(u)$, which means $d(u) \le |C| / (1-\epsilon)$.
\end{proof}

We then prove the following properties:

\begin{lemma}
For any $uv$ we have that both $OPT'_2(uv), OPT'(uv) \in \widetilde{D}$.
\end{lemma}
\begin{proof}
$OPT'_2(uv) \in \widetilde{D}$ follows directly by construction of $OPT'_2$.
$OPT'$ is obtained by modifying $OPT'_2$ without introducing any distances not in $OPT'_2$.
\end{proof}

\begin{lemma} \label{lem:subsetImportant}
Every non-singleton cluster in $OPT'_1, OPT'$ at level $w$ is an important cluster of $(V,E_w)$.
Every non-singleton cluster in $OPT'_2$ at level $w$ is a subset of some important cluster of $(V,E_w)$.
\end{lemma}
\begin{proof}
Every non-singleton cluster in $OPT'_1$ or in $OPT'$ is an important cluster (not just a subset of one), directly by \cite{cohen2022fitting} (Lemma 3.3, using $\epsilon$ instead of $\epsilon/8$).

As $OPT'_2$ is only splitting clusters of $OPT'_1$, the claim follows.
\end{proof}

It holds that all described trees are $O(1)$ approximations of $OPT$.

\begin{lemma} \label{lem:costOPTp}
$\|OPT'-D\|_0 = O(\|OPT-D\|_0)$.
\end{lemma}
\begin{proof}
By \Cref{lem:transformationVincent}, it suffices to show that $\|OPT'_2-D\|_0 = O(\|OPT'_1-D\|_0)$.

Let $C$ at level $w$ be a cluster of $OPT'_1$ that we modify in $OPT'_2$ (therefore $w\not \in \widetilde{D}$, by construction).

We first prove that there are no two non-singleton clusters $C_1,C_2$ at level $\widecheck{w}$ in $OPT'_1$ such that $C_1\subseteq C, C_2\subseteq C$.
Assume for the sake of contradiction that there exist such $C_1,C_2$, and w.l.o.g. $|C_1| \le |C_2|$.
By \Cref{lem:subsetImportant} we have that $C_1,C_2,C$ are all important clusters in their respective levels.
For any $u\in C_1$ we have $d_{\widecheck{w}}(u) \le |C_1| / (1-\epsilon)$, by \Cref{lem:importantDegrees}.
As $u\in C$, again by \Cref{lem:importantDegrees} we have $d_{w}(u) \ge (1-\epsilon)|C| \ge 2(1-\epsilon)|C_1|$.
But as $w\not \in \widetilde{D}$, we have that $d_w(u) \le (1+\delta)d_{\widecheck{w}}(u)$.
But then it should be $d_{\widecheck{w}}(u) \ge 2(1-\epsilon)|C_1|/(1+\delta)$ and at the same time $d_{\widecheck{w}}(u) \le |C_1| / (1-\epsilon)$, which is a contradiction for sufficiently small $\delta,\epsilon$.

Therefore, the only difference between $OPT'_1, OPT'_2$ is that certain nodes become singletons at some consecutive levels $w_1 > \ldots > w_k \not \in \widetilde{D}$ of $OPT'_2$ (for which $\widecheck{w_1} = \ldots = \widecheck{w_k}$), while they were already singletons at level $\widecheck{w_1}$ in $OPT'_1$.
Thus, for pairs including any such vertex $u$, the cost of $OPT'_2$ is increased by at most the number of outgoing edges of $u$ in these levels, that is by $x=|\bigcup_{i=1}^{k} N_{w_i}(u)|$.
But as $w_1 > \ldots > w_k$, we have $N_{w_1}(u) \supseteq \ldots \supseteq N_{w_k}(u)$, therefore $x=d_{w_1}(u) = O(d_{\widecheck{w_1}}(u))$ (due to $w_1\not \in \widetilde{D}$). 
However $u$ was already a singleton at level $\widecheck{w_1}$ of $OPT'_1$, and thus $OPT'_1$ was paying 
$d_{\widecheck{w_1}}(u)$ for pairs including $u$.
This proves our claim.
\end{proof}

We now prove some structural properties of $T$ related to $OPT'$.
Informally:
\begin{itemize}
    \item For every cluster $C$ of $OPT'$, there exists a cluster $C_T$ of $T$ at the same level.
    \item No cluster $C_T$ of $T$ contains two non-singleton clusters of $OPT'$ of the same level.
    \item Every cluster of $T$ is dense inside.
\end{itemize}

\begin{lemma} \label{lem:inclusion}
Let $C$ be a cluster of $OPT'$. Then there exists a cluster $C' \supseteq C$ of $T$ at the same level.
\end{lemma}
\begin{proof}
Let $OPT'_{sub}$ be a subtree of $OPT'$, whose root corresponds to a cluster $A\subseteq V$ and is at level $w$.
We prove an even stronger statement, namely that if we run \Cref{alg:HierL0} with parameters $(S,w)$ and obtain $T_{sub}$, where $S\supseteq A$, then for any cluster $C$ of $OPT'_{sub}$ there exists a cluster $C' \supseteq C$ of $T_{sub}$ at the same level.
This immediately implies the lemma, by setting $OPT'_{sub}=OPT'$ and $T_{sub}=T$ by running \Cref{alg:HierL0} with parameters $(V,w_{max})$.

The claim immediately follows if $|C|=1$.
It also follows for the topmost cluster of $OPT'_{sub}$, as it corresponds to $A$ while the root of $T_{sub}$ at the same level corresponds to $S\supseteq A$.

Now assume $C$ is a cluster of $OPT'_{sub}$, and let $C_p$ be its parent cluster. Inductively, $C_p$ is a subset of some cluster $C'_p$ of $T_{sub}$ at the same level.

If $C'_p$ is a core cluster, then $C$ is a subset of an important cluster by \Cref{lem:subsetImportant}.
As we obtain the children of $C'_p$ by $C'_p$-Structural Clustering, we obtain a cluster $C'$ containing the important cluster.

If $C'_p$ is not a core cluster, we find a set $R\subseteq C'_p$, create cluster $C'_p\setminus R$ in $T_{sub}$ at level $w$, and recurse on $R$ at level $\widehat{w}$.
Notice that, by Line~\ref{line:L0R}, $R$ only contains vertices $u$ with $\widetilde{d_{w''}(u)} < 0.65 |S|$ (where $S$ is such that $0.99|S| < |C'_p| \le |S|$), and there are at most $0.01|S|$ such vertices (Line~\ref{line:whileLoopL0}).
Therefore $0.98|S| < |C'_p\setminus R| \le |S|$.

If $u\in R$, then $d_w(u) < \frac{1+\lambda}{1-\lambda} 0.65|S| < 0.653 |S|$ for sufficiently small $\lambda$.
Similarly, at least $0.99|S|$ vertices in $C'_p$ have degree larger than $0.657|S|$ at level $w$ (Line~\ref{line:whileLoopL0}).

Now if $|C| > 0.01|S|$, then it contains some vertex with degree larger than $0.657|S|$ at level $w$; as $C$ is an important cluster (\Cref{lem:subsetImportant}), all vertices inside it have degree above $0.653 |S|$ (\Cref{lem:importantDegrees}), and therefore completely lies in $C'_p\setminus R$.
If $|C| \le 0.01|S|$, then again by \Cref{lem:importantDegrees} it can only contain vertices with degree at most $0.02|S|$, therefore only contains vertices in $R$; as we recurse on $R$, we inductively prove our claim.
\end{proof}

\begin{lemma} \label{lem:Tdense}
Let $C$ be a non-singleton cluster of $OPT'$ at level $w$.
Then any $u\in C$ has $d_w(u) > 0.6|C|$.
\end{lemma}
\begin{proof}
If $C$ is obtained by $S$-Structural Clustering, it directly follows that $u$ has $d_w(u) > 2|C|/3 > 0.6|C|$.
Otherwise, $C$ is obtained by removing all vertices with $\widetilde{d_w(u)} < 0.65|S|$ from its parent cluster $C_p$, for which we have $C_p \subseteq S$ for some vertex-set $S$.
But then all vertices in $C$ have $\widetilde{d_w(u)} \ge 0.65|S|$, which means $d_w(u) \ge \frac{1-\lambda}{1+\lambda} 0.65 |S| \ge \frac{1-\lambda}{1+\lambda} 0.65 |C|$, which implies our claim for sufficiently small $\lambda$.
\end{proof}

\begin{lemma} \label{lem:no2inside}
Let $C_1, C_2$ be non-singleton clusters of $OPT'$ at level $w$. There is no cluster $C'$ of $T$ at level $w$ such that $C'\supseteq C_1\cup C_2$.
\end{lemma}
\begin{proof}
Assume for the sake of contradiction that this is not true, and that w.l.o.g. $|C_1| \le |C_2|$.
By \Cref{lem:Tdense} any vertex in $C_1$ has degree at least $0.6|C|$ at level $w$.
But by \Cref{lem:importantDegrees} it has degree at most $|C_1| / (1-\epsilon) \le 0.5|C| / (1-\epsilon)$, which is a contradiction for sufficiently small $\epsilon$.
\end{proof}

We are now ready to prove that $T$ is a constant factor approximation of $OPT$.
\begin{lemma}
$\|T-D\|_0 = O(\|OPT-D\|_0)$.
\end{lemma}
\begin{proof}
By \Cref{lem:costOPTp}, it suffices to show that $\|T-D\|_0 = O(\|OPT'-D\|_0)$.
Let $E$ be the pairs $uv$ for which $OPT'(uv) = D(uv) \ne T(uv)$.
For all other pairs $T$ pays at most as much as $OPT'$.
In turn, it suffices to show $|E| = O(\|OPT'-D\|_0)$.

Let $uv\in E$.
By \Cref{lem:inclusion}, there exists a top level such that $u,v$ are in the same cluster $C$ in $T$ but not in $OPT'$.
By $OPT'(uv)=D(uv)\ne T(uv)$ we have that at this level $u$ and $v$ do not share an edge.
By \Cref{lem:no2inside}, one of the two nodes is a singleton in $OPT'$.
Let $e(u)$ be the degree of $u$ at the topmost level for which $u$ is a singleton in $OPT'$ but not in $T$ ($e(u)=0$ if no such level exists), and $c(u)$ be the cluster of $T$ containing $u$ in this level ($c(u)=\emptyset$ if no such level exists).
By the above discussion, $|E| \le \sum_{u\in V} |c(u)|$.

Notice that when $u$ is a singleton in $OPT'$ but not in $T$ (as $u$ is in a non-singleton cluster $C'$ of $T$ at that level), then $u$ has degree at least $0.6|C'|$, by \Cref{lem:Tdense}.
Therefore $\|OPT'-D\|_0 = \Omega(\sum_{u\in V} |c(u)|)$, which proves our claim.
\end{proof}

\begin{restatable}{lemma}{blackBoxes}\label{lem:blackBoxes}
Assume that within a single pass in the semi-streaming model, we can:
\begin{itemize}
    \item store a compressed set $\widetilde{D}$ of size $\widetilde{O}(n)$,
    \item store information of size $\widetilde{O}(n)$ that allows us to compute a $\widetilde{d_w(u)}$, for any vertex $u$ and weight $w\in \widetilde{D}$.
    \item store information of size $\widetilde{O}(n)$ that allows us to compute $S_i$-Structural Clustering for $k$ instances $\set{(V, E_{w_1}),S_1}, \ldots \set{(V, E_{w_k}), S_k}$.
    Instance $\set{(V, E_{w_i}),S_i}$ is only revealed after we compute $S_j$-Structural Clustering for every instance $\set{(V, E_{w_1}),S_j}$ with $j<i$ and may in fact depend on all these instances and the $S_j$-Structural Clusterings we output.
    Further, it holds that $w_i\in \widetilde{D}$ for all $i$, and each vertex $u\in V$ is contained in $O(\log{n})$ of all $S_i$.
\end{itemize}
Then we can $O(1)$-approximate \lo{} in a single pass in the semi-streaming model.
\end{restatable}
\begin{proof}
We simply run \Cref{alg:HierL0}.

Instead of explicitly storing the distances between every pair of vertices, we build a tree that induces these distances, in a top-down fashion.
At any given point, each leaf in the tree is associated with a subset of $V$, such that these subsets form a partition of $V$.
Initially we have a single node (the root) at height $w_{max}$, associated with $V$.
When we have $|S| = 1$, we simply create a leaf (corresponding to the unique node in $S$) at level $0$.

In Line~\ref{line:L0dist1}, we simply create $\mathcal{C}$ many children for the current node, each one corresponding to a different $C\in \mathcal{C}$, and recurse in all but the largest one, in case its size is larger than $0.99|S|$.
We only run $S$-Structural Clustering when $w\in \widetilde{D}$, and by \Cref{lem:recDepth} each vertex $u$ is only contained in $O(\log{n})$ $S_i$-Structural Clustering computations; therefore by assumption of the lemma, we can perform these $S_i$-Structural Clustering computations.

Similarly, in the while loop we have an active node (initially it is the unique cluster $C$ with $|C| > 0.99|S|$) at some level $w$.
Then we decide a set $R$ using the assumptions of our lemma, recurse on $R$ to create more children of our active node, and also create one more child associated with $C\setminus R$ at level $\widecheck{w}$.
Then we set the active node to be equal to the node corresponding to $C\setminus R$, and continue the execution.

It directly follows that the distances set in the algorithm are exactly the distances induced by our tree, and that the total space usage is $\widetilde{O}(n)$.
\end{proof}

We now prove our main theorem.
\lzerofit*
\begin{proof}
It suffices to guarantee that with high probability we can store the information required by \Cref{lem:blackBoxes}.
Our algorithm stores sketches for each vertex, as described in \Cref{section:sketches}, in a single pass.
In fact, it stores $c \log n$ independent instances of these sketches, for a sufficiently large $c$.
This requires $\widetilde{O}(n)$ space.

To compute $S_i$-Structural Clustering for $k$ instances $\set{(V, E_{w_1}),S_1}, \ldots \set{(V, E_{w_k}), S_k}$ such that $w_i\in \widetilde{D}$ and each vertex $u\in V$ is present in $O(\log{n})$ of all $S_i$, we employ \Cref{lem:horizontal}. 
In particular, using only the sketches of vertices in $S$, we can compute $S$-Structural Clustering with high probability.
As we store $c \log n$ independent sketches for each vertex, we can use different sketches for each computation.
Note that we cannot reuse our sketches, due to the dependencies across the instances and the clusterings we output.

We now show how to compute $\widetilde{d_w(u)}$ (which approximates $d_w(u)$), for any vertex $u\in V$ and weight $w\in D$ (this is stronger than $w\in \widetilde{D}$ required by \Cref{lem:blackBoxes}).
If $d_w(u) < 2\log^4 n$, then we can exactly compute it, as we explicitly store the $2log^4 n$ nearest neighbors of $u$.
Otherwise, by \Cref{clm:sketch1nb2}, we can report a sketch 
associated with size $s$ and weight $\tilde{w}^u_s$, such that with high probability, 
$\frac{d_{\tilde{w}^u_{s}}(u)}{1+5\zeta} \leq d_{w}(u) \leq \frac{d_{\tilde{w}^u_s}(u)}{1-\zeta}$.
Therefore, for sufficiently small $\zeta$, we have a sufficient approximation of $d_w(u)$.

Finally, we obtain $\widetilde{D}$ by using all the weights stored in memory; it follows its size is $\widetilde{O}(n)$.
To show that $\widetilde{D}$ is a compressed set with high probability, assume for the sake of contradiction that there exists a $w \not \in \widetilde{D}$ and a vertex $u$ such that $d_w(u) > (1+\delta) d_{\widecheck{w}}(u)$.
But as we proved in the previous paragraph, we have a weight $w'\in \widetilde{D}$ such that $d_{w}(u) \le \frac{d_{w'}(u)}{1-\zeta}$.
For $\zeta$ sufficiently smaller than $\delta$, this implies that $w' > w$.
For these to hold simultaneously, it must be that in the obtained sketch we only have vertices $v$ with $D(uv) \le \widecheck{w}$ or $D(uv)>w$.

It suffices to show that there exist $\nu d_{w'}(u)$ vertices $v \in N_{w'}(u)$ with $D(uv) \in (\widecheck{w},w]$, for a sufficiently small constant $\nu$.
This is because, if this is true, then with high probability we sample at least one such vertex.
On one hand, we have that at most $5\zeta d_{w'}(u)$ vertices with distance above $w$.
On the other hand he have at least $\delta d_{\widecheck{w}}(u)$ vertices with distance above $\widecheck{w}$.
If $\delta \widecheck{w} > 6\zeta d_{w'}(u)$, then we have at least $\zeta d_w{w'}(u)$ vertices with distance in $(\widecheck{w},w]$.
Otherwise we have $\widecheck{w} \le 6\zeta d_{w'}(u) / \delta$, meaning there are at least $(1-5\zeta - 6\zeta d_{w'}(u) / \delta)$ vertices with distance in $(\widecheck{w},w]$, for sufficiently small $\zeta$.
\end{proof}

\subsection{Lower bounds} \label{section:LowerBounds}

Lower bounds for the problem of correlation clustering in data streams were thoroughly examined in~\cite{ahn2021correlation}. In this section, we add additional natural results on top of this work.

The main lower bounds proved in~\cite{ahn2021correlation} were to the problem of testing if a given graph can be partition to clusters with optimal cost of $0$, under various edge weighting schemes. This observation leads directly to a similar computational limitation for algorithms that merely verify whether a matrix is ultrametric.
                          
\begin{theorem}
    Any randomized $k$-pass streaming algorithm that tests whether an input matrix is an ultrametric with probability greater than $\frac{2}{3}$ requires $\Omega(\frac{n}{k})$ bits.
\end{theorem}

\begin{proof}
Follows directly from Theorem~15 in~\cite{ahn2021correlation}.
\end{proof}

The subsequent theorems have implications for the problem of correlation clustering in streaming settings. We show that any algorithm addressing the correlation clustering problem, whether aiming to produce the optimal clustering or merely to report the optimal score, requires the use of $\Omega (n^2)$ bits. This requirement holds true even if the algorithm is permitted unbounded running time over the input. This results then naturally translate to the ultrametric construction framework.

\begin{proposition}\label{proposition:lbzeroclustering}
    Any randomized one-pass streaming algorithm that solves the correlation clustering problem with probability greater than $\frac{2}{3}$ requires $\Omega(n^2)$ bits.
\end{proposition}

\begin{proof}

    To prove the theorem we show a reduction to the index problem, where Alice is given a random string $x\in \{0,1\} ^{{\binom{n}{2}}}$ and Bob is given a random index $(i,j) \in {\binom{n}{2}}$ and they need to output $x_{i,j}$ using a one-way protocol from Alice to Bob. This problem is known to require $\Omega(n^2)$ bits of communication even for randomized protocols~\cite{ablayev1996lower}.

    Consider a protocol for the index problem where Alice exploits a one-pass algorithm for the correlation clustering problem and stream the edges, where the positive edges are $\{ (u,v) \mid x_{uv}=1 \}$. After the vector $x$ has been processed by Alice, Alice sends the content of the memory of the streaming algorithm to Bob.

    Bob stream two artificial cliques $\clique{1}, \clique{2}$ of size $2n$ each, together with the $+$ edges connecting $i$ to $\clique{1}$ and $j$ to $\clique{2}$, i.e. all edges $(i,c_1), (j,c_2)$ for $c_1 \in \clique{1}$ and $c_2 \in \clique{2}$.
    Furthermore, Bob stream $-$ edges between $C_1,C_2$ and the remaining vertices in the graph; that is, $(c,k)$ for $c \in \clique{1} \cup \clique{2}$ and $k \in [n] \setminus \{i,j\}$. 
    Finally, Bob stream exactly $\frac{(2n+1)^2-1}{2}$ $+$ edges between $\clique{1}$ and $\clique{2}$ and set the remaining edges to $-$ edges.

    Now, $x_{i,j}=1$ if and only if more than half of the edges between $\clique{1} \cup \{ i \}$ and $\clique{2} \cup \{ j \}$ are positive. It means an exact solution would have the cluster $\clique{1} \cup \clique{2} \cup \{ i,j \}$. That is, only if $i$ and $j$ are end up in the same cluster in the correlation clustering solution.
\end{proof}

As we will see in the following theorem, the space constraints remains also in the setting where the algorithm simply opt to report the cost of the clustering.

\begin{proposition}\label{proposition:lbzeroscore}
    Any randomized one-pass streaming algorithm that maintains the cost of an optimal correlation clustering solution with probability greater than $\frac{2}{3}$ requires $\Omega(n^2)$ bits.
\end{proposition}

\begin{proof}
    We again use a reduction from the index problem in a similar fashion, however we change the way Bob treats Alice's message.

    Bob first duplicate Alice's message and stream different information to each message.

    For the first message, Bob create 2 artificial cliques $\clique{i}, \clique{j}$ of size $n$ and connect them solely to $i,j$, respectively. Consequently, Bob obtains the cost of the optimal clustering. Due to the size of $\clique{i}, \clique{j}$, in any optimal clustering $i,j$ are in their newly added cliques. 

    For the second message, Bob create a single artificial clique $\clique{i,j}$ of size $2n$ and connect it to both $i,j$. Again, due to the size of $\clique{i,j}$ any optimal clustering contains $i,j$ in a cluster including $\clique{i,j}$.
    The cost is changed from the first message depending on whether $i$ is connected to $j$, namely,
    The cost will decrease by $1$ if and only if $i$ and $j$ are connected.     
    Hence, by subtracting the cost of the first message from the cost of the second message Bob gets an indicator stating if $x_{i,j}=1$ w.h.p.
    \end{proof}

We conclude these results in the following theorem:

\lbzeroclusteringorscore*

Fitting an ultrametric to a similarity matrix that contain just two specific values for under the $\ell_0$ or $\ell_1$ norms is exactly the correlation clustering problem (cf.~\cite{charikar}). It follows that the above bounds also holds for fitting ultrametric for both $\ell_0$ and $\ell_1$.

\lbzeroultrametric*

\section{\texorpdfstring{$\ell_\infty$ Ultrametrics}{l-inf Ultrametrics}} \label{section:linf}

In this section we provide a complete characterization of $\ell_\infty$ Best-Fit Ultrametrics in the semi-streaming model.
We show that in a single round, this problem cannot be approximated with an approximation factor strictly smaller than 2, while a factor 2-approximation algorithm in a single round does exist. Finally, we show that in two rounds we can obtain an exact solution. 

The lower bound result is derived from a reduction to the index problem in communication complexity. For the algorithmic results, we employ a reduction to the $\ell_\infty$ Min-Decrement problem, where we are only allowed to decrement the entries in the input matrix.

\subsection{\texorpdfstring{$\ell_\infty$ Ultrametrics lower bound}{l-inf Ultrametrics lower bound}}

\newcommand{\infnorm}[1]{\lVert #1 \rVert_{\infty}}

\linftylowerbound*

\begin{proof}
    To prove the theorem, we show a reduction to the index problem, where Alice is given a random string $x\in \{0,1\} ^{{ \binom{n}{2}}}$ and Bob is given a random index $(i,j) \in {\binom{\{n \}}{2}}$ and they need to output $x_{i,j}$ using a one-way protocol from Alice to Bob. This problem is known to require $\Omega(n^2)$ bits of communication even for randomized protocols~\cite{ablayev1996lower}.

    Assuming such an algorithm to $\ell_\infty$ Best-Fit Ultrametrics, Alice streams the matrix $\inputmatrix$ where $\inputmatrix (a,b)=x_{a,b}+1$ for every $(a,b) \in {\binom{\{ n \}}{2}}$.

    Bob, equipped with the index $(i,j)$, streams $\inputmatrix(n+1,i)= \inputmatrix (n+1,j)=0$ and $\inputmatrix (n+1,k)=1.5$ for $k \in [n] \setminus \{ i,j \}$, and obtains the output.

    Now, in the case that $x_{i,j}=0$, consider the matrix $\bar{O}$ with $\bar{O} (n+1,i)=\bar{O} ({n+1,j})=\bar{O} ({i,j})=0.5$ and $\bar{O} ({a,b})=1.5$ for all other indices.
    It is easy to verify that this is an ultrametric and that $\infnorm{\inputmatrix-\bar{O}} = 0.5$, so $OPT \leq 0.5$.

    However, in the case that $x_{i,j}=1$ we will show that $OPT=1$.
    Consider the matrix $\bar{O}$ with $\bar{O} (n+1,i)=\bar{O} ({n+1,j})=\bar{O} ({i,j})=1$ and $\bar{O} ({a,b})=1.5$ for all other indices.
    Similarly to the previous case, this is an ultrametric and $\infnorm{\inputmatrix-\bar{O}} = 1$, so $OPT \leq 1$. Let $O$ be some optimal solution, thus $\inputmatrix ({a,b})-OPT \leq O ({a,b}) \leq \inputmatrix ({a,b})+OPT$ for every $(a,b) \in {\binom{\{n \}}{2}}$. Combining this with the ultrametric property of $O$, we get:
    \begin{align*}
    2 - \text{OPT} &= \inputmatrix (i,j) - \text{OPT} \leq O (i,j) \leq \max \{ O ({n+1,i}), O ({n+1,j}) \} \\
    &\leq \max \{ \inputmatrix ({n+1,i}) + \text{OPT}, \inputmatrix ({n+1,j}) + \text{OPT} \} \leq \text{OPT},
    \end{align*}
    Consequently $OPT=1$. It follows that any randomized one-pass algorithm for $\ell_\infty$ Best-Fit Ultrametrics that claims an approximation factor strictly less than 2 would be capable of distinguishing between the two cases and correctly retrieving $x_{i,j}$ with good probability.
\end{proof}

\subsection{\texorpdfstring{$\ell_\infty$ Ultrametrics algorithms}{l-inf Ultrametrics algorithms}}

\newcommand{\infopttree}{O}
\newcommand{\infmdopttree}{\bar{O}}
\newcommand{\maxoftwo}[2]{\max \{ #1, #2 \}}
\newcommand{\optcost}{c}

To solve $\ell_\infty$ Best-Fit Ultrametrics we will apply a reduction to the $\ell_\infty$ Min-Decrement Ultrametrics problem. In this variant, it is only allowed to decrement the entries in the input matrix. We will show that an optimal solution to this variant is 2-approximation to the best fit.
\begin{lemma}\label{lemma:md_approximation}
An optimal solution to the $\ell_\infty$ Min-Decrement Ultrametrics problem is at most 2 approximation to $\ell_\infty$ Best-Fit Ultrametrics.
\end{lemma}

\begin{proof}
Let $\infopttree$ denote an optimal solution to $\ell_\infty$ Best-Fit Ultrametrics given the input matrix $\inputmatrix$, where the optimal fitting cost $\infnorm{\infopttree - \inputmatrix}$ is denoted by $\optcost$. 
Consequently, we have $\inputmatrix \geq \infopttree - \optcost$, now set $\infmdopttree ({i,j}) = \max \{ 0, \infopttree ({i,j}) - \optcost \}$, it follows that  $\infmdopttree \leq \inputmatrix$.

Note that $\infmdopttree$ is also an ultrametric, for every $i,j,k$:

\begin{align*}
\infmdopttree ({i,j}) &= \maxoftwo{0}{\infopttree ({i,j}) -\optcost} \leq \maxoftwo{0}{\maxoftwo{\infopttree ({i,k})}{\infopttree ({k,j})} - \optcost} \\
&= \maxoftwo{\maxoftwo{0}{\infopttree ({i,k}) - \optcost}}{\maxoftwo{0}{\infopttree ({k,j}) - \optcost}} = \maxoftwo{\infmdopttree ({i,k})}{\infmdopttree ({k,j})}
\end{align*}

Additionally, $\infmdopttree$ is a 2-approximation of $\infopttree$: 
\begin{align*}
\infnorm{\inputmatrix-\infmdopttree} &= \infnorm{\inputmatrix-\maxoftwo{0}{\infopttree-\optcost}} = \infnorm{\inputmatrix+\min \{0, -\infopttree+\optcost \} } \\
&= \infnorm{ \min \{ \inputmatrix, \inputmatrix-\infopttree + \optcost \} } 
=\max_{i,j} \{ \min \{ \inputmatrix (i,j),\inputmatrix (i,j)-O (i,j) +c \} \} \\
&\leq \max_{i,j} \inputmatrix (i,j) - O (i,j)+c
\leq 2\optcost
\end{align*}
\end{proof}

Next, we will show that the  $\ell_\infty$ Min-Decrement Ultrametrics can be derived from the Minimum Spanning Tree (MST). Similar ideas were used in Theorem~3.3 in~\cite{agarwala}.

Given an input matrix $\inputmatrix$, let $T$ be an MST of $\inputmatrix$.
$T$ naturally yields an ultrametric by defining the distance between any two vertices $i$ and $j$ as the weight of the heaviest edge on the unique path connecting them, denoted henceforth by $T ({i,j})$.
This construction inherently satisfies the ultrametric property, as the tree structure ensures that there is exactly one path between any pair of vertices, thereby maintaining the ultrametric property.

Moreover, if the edge $(i,j)$ of weight $\inputmatrix (i,j)$ is in $T$, then clearly $T ({i,j})=\inputmatrix (i,j)$. If not, the edge forms a cycle with the edges of $T$. Given that $T$ is an MST, it follows that $T ({i,j}) \leq \inputmatrix (i,j)$. Therefore, $T$ is indeed a minimum decrement ultrametric of $\inputmatrix$.

Furthermore, every minimum decrement ultrametric has values smaller or equal to the values of $T$. 
To see this let $\infmdopttree$ be an optimal minimum decrement ultrametric. For any edge $(i,j)$, if $(i,j)$ belongs to $T$ then $T ({i,j})=\inputmatrix (i,j)$ and since $\infmdopttree \leq \inputmatrix$ it follows that $\infmdopttree ({i,j}) \leq T ({i,j})$. Else, let $P$ be the path from $i$ to $j$ in $T$. Due to the ultrametric property, $\infmdopttree ({i,j}) \leq \max_{(k,l)\in P} \infmdopttree ({k,l}) \leq 
\max_{(k,l)\in P} \inputmatrix ({k,l}) = \max_{(k,l)\in P} T ({k,l}) = T ({i,j})$.

Therefore, the minimum spanning tree provides an optimal solution to the minimum decrement problem. 
As noted in~\cite{feigenbaum2005graph}, the minimum spanning tree can be constructed in a single pass with $O(\log n)$ time per edge under the semi-streaming model. We conclude this result in the following lemma:

\begin{lemma}\label{lemma:optimal_md}
An optimal solution to the $\ell_\infty$ Min-Decrement Ultrametrics fitting problem can be constructed in a single pass over the stream with $O(\log n)$ time per edge.
\end{lemma}

As proved in Proposition~\ref{theorem:l_infty_lower_bound},
this construction achieves the best possible approximation within a single pass over the stream.
The next theorem is now an immediate consequence of Lemma~\ref{lemma:md_approximation} and Lemma~\ref{lemma:optimal_md}.

\inftysinglepass*

We proceed to demonstrate that a second pass over the stream, while necessary, is also sufficient to achieve the optimal solution to $\ell_\infty$ Best-Fit Ultrametrics.

The implementation goes by first applying Theorem~\ref{theorem:theorem_single_pass_md} to produce an ultrametric $T$. Then, in a second pass over the stream, simply compute the error of $T$ on the input $\inputmatrix$, denoted by $\bar{\optcost}$, and return $T'=T+\frac{\bar{\optcost}}{2}$.
The error of $T'$ is $\frac{\bar{\optcost}}{2}$, as $0 \leq \inputmatrix-T \leq \bar{\optcost}$ it follows that $-\frac{\bar{\optcost}}{2} \leq \inputmatrix -T -\frac{\bar{\optcost}}{2} \leq \frac{\bar{\optcost}}{2}$.

According to Lemma~\ref{lemma:md_approximation}, $\bar{\optcost}$ is at most twice the optimal cost. That is, the ultrametric $T'=T+\frac{\bar{\optcost}}{2}$ achieves the optimal cost precisely.
It also follows that the MST obtained from Theorem~\ref{theorem:theorem_single_pass_md} provides precisely a 2-approximation of the optimal ultrametric and at the same time has the topology of the optimal solution.
This now fully concludes the $\ell_\infty$ Best-Fit Ultrametrics problem in the streaming settings.

\twopass*

\section{\texorpdfstring{$\ell_0$ and $\ell_\infty$ Tree Metrics}{l-0 and l-inf Tree Metrics}} \label{section:trees}
\newcommand{\apivot}{a}
\newcommand{\centroid}{C^\apivot}
\newcommand{\zeronorm}[1]{\lVert #1 \rVert_0}

The problem of $\ell_p$ Best-Fit Tree-Metrics is typically addressed through reduction to an $\ell_p$ Best-Fit Ultrametrics instance, introducing a constant multiplicative approximation factor. This reduction, first introduced for the $\ell_\infty$ norm, generalizes to every $\ell_p$ with $p \geq 1$~\cite{agarwala}.
More recently, Kipouridis showed how to adapt this reduction to the $\ell_0$ case as well~\cite{kipouridis2023fitting}.

In this section we show how this methodology can be extended to the semi-streaming model. We will show that with just an additional pass over the input stream it is possible to construct the best fit tree metric. 
In what follows we will focus on $\ell_0$ and $\ell_\infty$, aligning with the algorithms proposed in this paper. However, this method can be generalized for any $\ell_p$ norm with $p \geq 1$.

The reduction strategy involves selecting a pivot element $\apivot$, for which let $\centroid$ denote the centroid metric defined by $\centroid (i,j) = 2\max_{k \in [n]}\inputmatrix (\apivot,k) - (\inputmatrix (\apivot,i) + \inputmatrix (\apivot,j))$. Using the best fit ultrametric algorithm we then obtain an ultrametric $U^\apivot$ of $\inputmatrix + \centroid$ and an $\apivot$-restricted tree of $\inputmatrix$ by setting $T^\apivot = U^\apivot - \centroid$.
Where $A$ is denoted an $\apivot$-restricted metric of $B$ if, $A (\apivot,k)=B (\apivot,k)$ for every $k \in [n]$ (in this context $A,B$ are symmetric matrices). We will see that $T^\apivot$ is a constant approximation to the best fit tree metric.

Note that if all the values $\inputmatrix(\apivot) \coloneqq \left(\inputmatrix (\apivot,k) \right)_{k \in [n]}$ are stored in memory, it is possible to adjust the input distance matrix $\inputmatrix$ as the stream is processed and ultimately compute $T^\apivot$ in a single pass. Consequently, the first pass is utilized for storing $\inputmatrix (\apivot)$ for some predefined $\apivot$, and the second pass computes the $\apivot$-restricted tree metric $T^\apivot$ as outlined above.

Using this idea we will show how to solve the problem of tree metric fitting for both $\ell_0$ and $\ell_\infty$.

\subsection{\texorpdfstring{$\ell_\infty$ Best-Fit Tree Metrics}{l-infty Best-Fit Tree Metrics}} 
In the case of $\ell_\infty$, any arbitrary selection of a pivot $\apivot$ will provide with a constant approximation factor. 
Let $\tilde{U}^\apivot$ denote the 2-approximation ultrametric of $\inputmatrix+\centroid$ obtained by the algorithm outlined in Theorem~\ref{theorem:theorem_single_pass_md}. Recall that this is a minimum decrement ultrametric.

We show the following lemma, similar ideas were also utilized in~\cite{agarwala}. 

\begin{lemma}\label{lemma:linf_tree_metric}
For every element $\apivot$, $T^\apivot = \tilde{U}^\apivot - \centroid$ is a 2-approximation $\apivot$-restricted tree metric.
\end{lemma}

\begin{proof}    
Let $m_a = \max_{k \in [n]}\inputmatrix (\apivot,k)$.
To see that $T^\apivot$ is $\apivot$-restricted we will show that for every $i$, $\tilde{U}^\apivot (a,i) = 2m_a$, it is then easy to verify that for every $i \in [n]$, $T^\apivot(a,i)=(\tilde{U}^\apivot - \centroid)(a,i) = \inputmatrix(a,i)$.
First note that, $(\inputmatrix + \centroid)(a,i) = 2m_a$. Then, since $\tilde{U}^\apivot$ is an MST of $\inputmatrix + \centroid$, we have $\tilde{U}^\apivot (a,i) = 2m_a$.

To show that $T^\apivot$ is a tree metric we will use the following claim as in~\cite{agarwala}.

\begin{claim}\label{claim:in_agrawal}
For every $\apivot \in [n]$, $T$ is a tree metric if and only if $T+\centroid$ is an ultrametric.
\end{claim}

From Claim~\ref{claim:in_agrawal} it immediately follows that $T^\apivot$ is a tree metric.
It is left to show that $T^\apivot$ is a 2-approximation to any optimal $\apivot$-restricted tree metric of $\inputmatrix$, denoted $T^{OPT}$.

\begin{align*}
\infnorm{T^{OPT}-\inputmatrix} &= \infnorm{(T^{OPT}+\centroid)-(\inputmatrix+\centroid)} \\
&\geq \infnorm{U^{OPT}-(\inputmatrix+\centroid)} \quad \text{(by Claim~\ref{claim:in_agrawal}, for some optimal ultrametric $U^{OPT}$)} \\
&\geq \frac{1}{2}\infnorm{\tilde{U}^\apivot-(\inputmatrix+\centroid)} = \frac{1}{2} \infnorm{T^\apivot - M}
\end{align*}

Overall, $\infnorm{T^\apivot - M} \leq 2\infnorm{T^{OPT}-\inputmatrix}$.
\end{proof}

Using Lemma~3.4 in~\cite{agarwala}, an optimal $\apivot$-restricted tree metric is 3-approximation of the optimal tree metric. Thus, as a consequence of Lemma~\ref{lemma:linf_tree_metric}, for any selection of pivot $\apivot$, the output is a 6-approximation tree metric to the optimal tree metric. We summarize this in the following theorem:

\inftyreduction* 

\subsection{\texorpdfstring{$\ell_0$ Best-Fit Tree Metrics}{l-0 Best-Fit Tree Metrics}} 

While in the $\ell_\infty$ case every selection of a pivot would result in a 6 approximation, this does not hold for $\ell_0$; yet, Kipouridis proved the existence of $\apivot \in [n]$ which achieves a 3-approximation. Kipouridis then executed $n$ reductions, that included best ultrametric fit, to obtain the desired approximation tree metric. 

\begin{lemma}[Theorem~3 in~\cite{kipouridis2023fitting}]\label{lemma:best_tree_0}
A factor $\rho \geq 1$ approximation for $\ell_0$ Fitting Ultrametrics implies a factor $6\rho$
approximation for $\ell_0$ Fitting Tree Metrics.
\end{lemma}

Since we cannot store every $\inputmatrix(\apivot)$ in memory we will have to suggest a different scheme.
We will show in the following lemma that for a randomly selected $\apivot \in [n]$, obtaining an optimal  $\apivot$-restricted tree is a constant approximation to the optimal best fit tree.

\begin{lemma}\label{lemma:l0_pivot}
If $\apivot$ is randomly selected then with probability $\geq \frac{3}{4}$ the resulting tree metric is at most $12$ approximation to the $\ell_0$ Best-Fit Tree-Metrics.
\end{lemma}

\begin{proof}
Let $T$ denote an optimal best fitting tree metric to $\inputmatrix$ under $\ell_0$.
We have that, $OPT=\zeronorm{\inputmatrix-T} = \frac{1}{2} \sum_{i\in [n]}\zeronorm{\inputmatrix(i)-T(i)}$.
Fix $i \in [n]$, we will transform T to an $i$-restricted tree, and denote this as $T^{/i}$, note that this is not necessarily an optimal $i$-restricted tree on $\inputmatrix$. The transformation works by moving every $j$ either toward or away from $i$ until each $j$ is at distance $\inputmatrix (i,j)$. This transformation is also described in Lemma~3.4 in~\cite{agarwala} and Theorem~3 in~\cite{kipouridis2023fitting}.

We next show that $T^{/i}$ is a good approximation of $T$. 
Observe that by moving $j$, only distances from/to $j$  may be modified, thus at most $n$ errors may be introduced. Moreover, $j$ is moved only if $\inputmatrix (i,j) \neq T (i,j)$, so the number of errors introduced by this transformation is at most $n \zeronorm{\inputmatrix(i)-T(i)}$.

Summing over all $i$ we get that:
\begin{align*}
\sum_{i\in [n]}\zeronorm{T^{/i}-\inputmatrix} &\leq 
\sum_{i\in [n]} OPT + n \zeronorm{\inputmatrix(i)-T(i)} \\
&\leq n\cdot OPT +n \sum_{i \in [n]}\zeronorm{\inputmatrix(i)-T(i)} \leq n\cdot OPT +n \cdot 2OPT = 3n \cdot OPT
\end{align*}

So by randomly selecting $\apivot \in [n]$ (with uniform distribution), the expected cost of an optimal $\apivot$-restricted tree metric is at most $3OPT$. Then, through Markov's inequality, the probability that $\zeronorm{T^{/\apivot} - \inputmatrix} > 12OPT$ is less than $\frac{1}{4}$.
\end{proof}

In order to further improve the algorithm and obtain a high probability success rate, we sample not one but $t=\ln n$ pivots, $P=\{\apivot_1,...,\apivot_t \}$, and obtain $t$ trees $\{T^{\apivot_1},...,T^{\apivot_t}\}$, each is at most 12 approximation to the optimal fit with probability at least $\frac{3}{4}$ following Lemma~\ref{lemma:l0_pivot}. 
Let $\overline{OPT}$ be the minimum value such that the graph $G_{\overline{OPT}}$ over the vertex set $P$, with edges $(\apivot_i,\apivot_j)$ where $\zeronorm{T^{\apivot_i}-T^{\apivot_j}} \leq 24\overline{OPT}$, has a clique of size at least $\frac{1}{2}n$. Finally, we arbitrarily select some pivot $\apivot_i$ in that clique and return $T^{\apivot_i}$.

Note that by definition, if $a<b$ then $G_a \subseteq G_b$. Hence, we can find $\overline{OPT}$ by carrying a binary search in the range of possible values of $OPT$, i.e. $\overline{OPT} \in [0,n^2]$. Since $t$ is logarithmic in $n$, this entire process can be carried in semi-streaming settings.

\begin{claim}
With high probability, any pivot selected from the outlined clique in $G_{\overline{OPT}}$ corresponds to at most 36 approximation of $T$.
\end{claim}

\begin{proof}
Consider $G_{OPT}$ and let $P'$ denote the set of all vertices in $P$ that correspond to a 12 approximation of $T$. We first show that $|P'|\geq \frac{t}{2}$ w.h.p.

For every $\apivot_i , \apivot_j \in P'$, 
$\zeronorm{T^{\apivot_i}-T^{\apivot_j}} \leq \zeronorm{T^{\apivot_i}-T} + \zeronorm{T^{\apivot_j}-T} \leq 24OPT$. It follows that there is an edge between every two vertices in $P'$. 

Write $X=\sum_{i=1}^t X_i$ where $X_i$ is the indicator that $\apivot_i \in P'$, and let $\mu = \mathbb{E} [X]$. Following Lemma~\ref{lemma:l0_pivot}, $\mu \geq \frac{3}{4}t$. 

Using Chernoff bound it holds that:
\begin{align*}
\mathbb{P} [X \leq (1-\delta) \mu] \leq \exp (\frac{-\delta^2 \mu}{2}) = (\frac{1}{t})^{3\delta^2/8}
\end{align*}
Let $\delta = \frac{1}{3}$ and obtain that w.h.p there is a clique in $G_{OPT}$ of size at least $\frac{1}{2}t$.

Recall that, if $a<b$ then $G_a \subseteq G_b$, that is, w.h.p, $\overline{OPT} \leq OPT$.
The probability that one of the vertices of the clique is at most 12 approximation of $T$ is at least $1-(1-\frac{3}{4})^{t/2}=1-(\frac{1}{2})^t$. By selecting any pivot $\apivot_i$ in that clique with a corresponding tree $T^{\apivot_i}$ we have that w.h.p we report an $\apivot_i$-restricted tree that is at most 36 approximation to $T$. Since w.h.p there exist $T^{\apivot_j}$ that is a 12 approximation to $T$, also,
$\zeronorm{T^{\apivot_i}-T^{\apivot_j}} \leq 24\overline{OPT} \leq 24OPT$. 
It Follows that:
\begin{align*}
\zeronorm{T^{\apivot_i}-T} \leq \zeronorm{T^{\apivot_i}-T^{\apivot_j}} + \zeronorm{T^{\apivot_j}-T} \leq 24OPT + 12OPT \leq 36OPT
\end{align*}
\end{proof}

Together with Lemma~\ref{lemma:best_tree_0} we obtain the theorem:

\zeroreduction*

\bibliographystyle{alphaurl}
\bibliography{correlation}

\end{document}